\documentclass[12pt]{article}
\usepackage{amsmath,amssymb,bm,amsthm,array,tikz,enumerate, braket, lineno }
\usepackage{booktabs}
\usepackage{graphicx}
\usetikzlibrary{graphs}
\tikzset{mynodes/.style={circle,white,fill=black,text width=.8em,inner sep=0pt,text centered,font=\footnotesize}} 
\tikzset{myedges/.style={thick}} 
\tikzset{graphs/every graph/.style={nodes={mynodes}, edges={myedges}}} 
\tikzset{graphs/digraph/.style={edges={myedges,-latex}}} 
\tikzset{every loop/.style={looseness=30}} 
\numberwithin{equation}{section}
\usetikzlibrary{positioning,graphs}

\DeclareMathOperator{\Spec}{Spec}

\DeclareMathOperator{\Ker}{Ker}

\newcommand{\bs}[1]{\boldsymbol{#1}}

\newtheorem{thm}{Theorem}[section]
\newtheorem{lem}[thm]{Lemma}
\newtheorem{pro}[thm]{Proposition}
\newtheorem{cor}[thm]{Corollary}

\newtheorem*{maintheorem1}{Theorem \ref{main}}
\newtheorem*{maintheorem2}{Theorem \ref{main2}}
\theoremstyle{definition}
\newtheorem{defi}[thm]{Definition}

\allowdisplaybreaks

\title{Robustness of
periodicity in Grover walks under a magnetic vector potential }
\author{Hiroto Sekido~$^{\dagger,}$\footnote{corresponding author: 
\texttt{sekido-hiroto-zk@ynu.jp}}
,\;\; Etsuo Segawa~$^\dagger$ \\
{\small ~$^\dagger$Graduate School of Environment and Information Sciences,}\\
{\small Yokohama National University, }\\
{\small Hodogaya, Yokohama 240-8501, Japan.}
}
\date{}

\begin{document}

\maketitle
\begin{abstract}
We study the effect of magnetic vector potentials on periodic Grover walks on finite graphs. The magnetic vector potential is introduced through the framework of quantum graphs, which induces the Grover walk as a special case. We regard the magnetic vector potential as a perturbation of a periodic Grover walk and investigate the robustness of its periodicity.
Our analysis reveals that the response to such perturbations depends on the spectral structure of the underlying graph. In particular, when the graph possesses at least one non-simple eigenvalue, we derive a Hermitian matrix that characterizes the robustness of its periodicity.  {As a consequence, for initial states orthogonal to the eigenspaces of the unperturbed Grover walk corresponding to the eigenvalues $\pm1$, we show that the perturbed dynamics is asymptotically described by a continuous-time quantum walk generated by this Hermitian matrix.}
\vspace{8pt}  
\end{abstract}

\section{Introduction}


Quantum walks are the quantum mechanical counterpart of classical random walks and have been extensively studied in connection with quantum computation~\cite{Am, JK, PR} and various related fields, e.g.,~\cite{Mo, Ya}. In particular, the discrete-time quantum walks on graphs  exhibit a variety of distinct behaviors depending on the underlying graph. Consequently, a considerable amount of research has been devoted to characterizing graph structures through quantum-walks phenomena such as perfect state transfer~\cite{CZ,SKEG,KSYY,MSSS,MSSS2}, periodicity~\cite{YNIE,SHH,KSYY,SHK,Y2019}, effectiveness of quantum search~\cite{BW1,HZ}, graph isomorphism detection~\cite{BW}, attainment of maximum entropy~\cite{SHE}, orientability detection of graph embeddings into closed surfaces~\cite{HS}, etc. The spectral structure of the time evolution operator plays a central role in connecting the behavior of a quantum walk with the underlying graph structure, e. g.~\cite{SMT,QGW1,QGW2}. Among the various characteristic phenomena of discrete-time quantum walks, this paper focuses on periodicity, for which the influence of the spectral structure appears particularly prominently.

Periodicity is defined by the property that the time evolution operator returns to the identity after a finite number of steps. This phenomenon is closely related to the fact that the spectral phases are rational multiples of \(2\pi\)~\cite{PR}. For the Grover walk, this spectral characterization has led to various results on periodicity and the classification of periodic graphs~\cite{YNIE,SHK,Y2019}. However, much less is known about the robustness of periodicity under perturbations.

One natural source of perturbations arises from quantum graph models, where magnetic vector potentials can be incorporated into quantum dynamics. {In this paper, the magnetic vector potential is also referred to simply as the ``vector potential.''} A quantum graph (e.g.,~\cite{PP, SU, TU, P}) is a metric graph equipped with differential operators acting on the edges together with boundary conditions imposed on the vertices. For a finite graph \(G=(V,E)\), each edge is identified with a finite interval, and a Schrödinger equation is defined in each interval. The solutions on the adjacent edges are coupled through boundary conditions at the vertices, and the resulting spectral problem can be described by a bond scattering matrix acting on the set of directed edges $A(G)$ obtained by replacing each edge of $G$ with a pair of oppositely directed edges.

Several studies have revealed close connections between quantum graphs and discrete-time quantum walks~\cite{QGW1,QGW2,SG,GK}.  The state space of this induced quantum walk is the Hilbert space $\mathbb{C}^{A(G)}$ with the standard inner product. They showed that stationary solutions of a quantum graph can be characterized in terms of eigenvectors associated with the eigenvalue $1$ of this corresponding quantum walk. Consequently, spectral properties of quantum graphs can be investigated within the framework of quantum walks. Furthermore, Grover walks can be realized as a special case of this quantum walk. In particular, magnetic vector potentials associated with directed edges appear naturally as phase factors in this induced quantum walk.

In this paper, we regard the magnetic vector potential as a perturbation of a periodic Grover walk and investigate the robustness of its periodicity. Our analysis reveals that the response to such perturbations depends crucially on whether the Grover walk has non-simple eigenvalues (from the proof of Lemma~\ref{key lemma}). { Grover walk always has an eigenvalue $1$ for any connected graph and has an eigenvalue $-1$ for any connected bipartite graph \cite{SMT}, whose spectral phases are rational multiples of $2\pi$ of course, and  the detection of the periodicity strongly depends on the rest of eigenvalues. Then, we focus on the effect of non-simple eigenvalues  except for $\pm1$ on the robustness of periodicity.} In particular, if the graph has at least one non-simple eigenvalue, {then, for initial states taken from the subspace $\mathcal{H}_0$ orthogonal to the eigenspaces of eigenvalues $\pm1$,} its dynamics is effectively described by a certain continuous-time quantum walk, which captures a gap from the original periodic behavior (from Theorem~\ref{main}). Recent studies, including \cite{Childs, SS, Strauch}, have investigated the approximations of discrete-time quantum walks by continuous-time quantum walks. Within this line of research, our work presents an approach based on periodicity, where a continuous-time quantum walk emerges as a description of a gap from periodic behavior induced by a magnetic flux.

The remainder of this paper is organized as follows. In Section 2, we introduce the necessary definitions and preliminary results on quantum walks and magnetic vector potentials. In Section 3, we introduce several definitions and state the main results in a formal way. In Section 4, we present several lemmas and prove the main results. Finally, in Section 5, we provide three examples.

\subsection{First main result}

The main result of this paper can be stated informally as follows. The first main result establishes a connection between periodic Grover walks with magnetic perturbations assigned to a fixed directed edge, called the reference edge, and continuous-time quantum walks. More precisely, we show that, {for the initial state orthogonal to the eigenspaces of eigenvalues $\pm1$,} under a suitable scaling limit, the perturbed periodic dynamics converges to a continuous-time quantum walk generated by a Hermitian matrix $H$. The detailed construction of $H$ is given in Definition~\ref{H1}.

\begin{maintheorem1}
    Let $G$ be a graph that induces a $\tau$-periodic Grover walk. For the small strength {$\beta$} of a vector potential and a reference edge $a \in A(G)$, denote $U_{\beta} = U_{\beta}(G,a)$ and $H = H(G,a)$.   Set $\phi_0  {\in \mathcal{H}_0}$ as the initial state, and $\varphi_{\lceil \frac{t}{\beta}\rceil}^{[D]}$   as the state at the final time 
    $\lceil \frac{t}{\beta}\rceil$ of the discrete-time quantum walk with some $t>0$, defined by the following time evolution:
\[
\varphi_0^{[D]} = \phi_0;\  \varphi_n^{[D]} = U^{\tau}_{\beta} \varphi_{n-1}^{[D]} \ 
\left(n=1,\dots, \left\lceil \frac{t}{\beta} \right \rceil
\right).
\]
     Set $\psi_t^{[C]}$ as the state at time 
    $t$ of the continuous-time quantum walk, defined by the following time evolution:
    {
    \begin{align*}
\psi_0^{[C]} = \phi_0;\  
i \frac{\partial}{\partial t}\psi_t^{[C]} = \tau H\psi_t^{[C]} \ (t>0).
\end{align*}
}
Then, we have
\[
\psi_t^{[C]} = \lim_{\beta \to 0} \varphi_{\lceil \frac{t}{\beta}\rceil}^{[D]}.
\]
\end{maintheorem1}

\subsection{Second main result}

The second main result concerns the spectral properties of $H$ and the robustness of  periodicity under magnetic vector potentials. The nonzero eigenvalues and the corresponding eigenvectors of $H$ can be expressed in terms of the eigenvalues and their associated eigenvectors of the discriminant matrix $T$. Here, $T$ is similar to the probability transition matrix of the isotropic random walk on $G$, and the explicit relation between them is given in \eqref{RWbyT}.
To formulate this notion precisely, we introduce the following definition of robustness of periodicity. 
\begin{defi}[ Robustness of periodicity ]
    Let $G=G(V, E)$ be a graph that induces periodic Grover walks. For a reference edge $a \in A= A(G)$, set  $H = H(G,a)$. The robustness of periodicity, denoted by $\mathcal{R}_p= \mathcal{R}_p(G, a)$, is defined as follows:
    \[
    \mathcal{R}_p = \frac{\dim \ker H}{|A|}.
    \]
\end{defi}
Recall that $\psi_t^{[C]}$ is denoted by 
\[
 i \frac{\partial}{\partial t}\psi_t^{[C]} = \tau H\psi_t^{[C]} 
\]
with $\psi_0^{[C]}=\varphi_0^{[D]} \in \mathcal{H}_0$, and according to Theorem~\ref{main}, $\psi_0^{[C]}$ is characterized by
\[
\psi_t^{[C]} = \lim_{\beta \to 0} \varphi_{\lceil \frac{t}{\beta}\rceil}^{[D]}= \lim_{\beta \to 0} \left(U_\beta^\tau\right)^{\lceil \frac{t}{\beta}\rceil}\varphi_0^{[D]}.
\]
Here, $U_0^\tau=I$. Then, ``$\partial \psi_t^{[C]}/\partial t\equiv0$", that is, $\psi_t^{[C]}\equiv \varphi_0^{[D]},$  shows that no influence from the perturbation appears in the induced continuous-time quantum walk. In this sense, we regard $\mathcal{R}_p = \dim \ker H/|A|$  as measuring the robustness of the periodic quantum walks to the perturbation in this paper.  



As we will see precisely later in Section~3,  the dimension of $\ker H$ and the spectrum of $H$ can be characterized in terms of how the discriminant $T : \mathbb{C}^V \to \mathbb{C}^V$ overlaps with the endpoints $\partial X_a$ and the number of fundamental cycles in $G$. 
More precisely, {as will be defined by \eqref{Csub},}  let $\mathcal{B}_{\partial X_a}$ denote the {subspace} of vectors supported on the end points $\partial X_a$ of a reference edge $a$ carrying the magnetic vector potential.
Since $|\partial X_a|=2$, we have
\[
\kappa
:=
\dim\!\left(\ker(T-\lambda_TI)\cap\mathcal B_{\partial X_a}\right)\le2
\]
for every eigenvalue $\lambda_T$ of the discriminant $T$.
This observation naturally classifies the eigenspaces of the discriminant according to
\[
\kappa \in \{0,1,2\},
\]
which forms the basis of our analysis.
The cases $\kappa \in \{0,1\}$ give rise to the subspaces 
$\mathcal{S}_{\mathrm{sim}}$ ($\kappa=1$), 
$\mathcal{T}_{\mathrm{per}} $ ($\kappa=0$) $\subset \ker H$, respectively, while the case $\kappa=2$
determines the non-zero spectrum of $H$. 
The precise definitions of $\mathcal{S}_{\mathrm{sim}}$ and
$\mathcal{T}_{\mathrm{per}}$ are given in Definition~\ref{ST}.
In addition, the subspace $\mathcal{L}^{\perp}$, which arises from the
fundamental cycles of the underlying graph~\cite{SMT}, constitutes the
remaining component of $\ker H$, that is, $\ker H = \mathcal{S}_{\mathrm{sim}} 
    \oplus
    \mathcal{T}_{\text{per}}
    \oplus
    \mathcal{L}^{\perp}$. 
For the case 
$\kappa=\dim\!\left(\ker(T-\lambda_T I)\cap\mathcal B_{\partial X_a}\right)=2$, 
the non-zero eigenvalues of $H$ induced by $\lambda_T$ are explicitly described in terms of 
the local quantity $E_{\lambda_T}(a,a)$, which admits a geometric interpretation as the area of the parallelogram generated by the two basis vectors of $\ker(T-\lambda_TI\cap \mathcal{B}_{\partial X_A})$ restricted to $\partial X_a$; and the eigenvectors of $H$ are described by the vectors lifted by the operator $\partial_\lambda^{(\lambda_T)}: \mathbb{C}^X\to \mathbb{C}^{A}$ from the above two basis vectors in $\mathbb{C}^X$.  Its precise definition is given in Definition~\ref{H1}.

\begin{maintheorem2}
    Set $U_{0} = U_{0}(G)$ and $H = H(G, a)$, for a graph $G$ that induces $\tau$-periodic Grover walks and a reference edge $a \in A(G)$. For  $\lambda_T  \in \Spec(T)$ with $\dim \left( \ker (T- \lambda_T I) \cap \mathcal{B}_{\partial X_a}\right) =2$, we denote by $f_1^{(\lambda_T)}, f_2^{(\lambda_T)}$ an orthonormal basis of $\ker (T- \lambda_T I) \cap \mathcal{B}_{\partial X_a}.$ Then, the spectrum of $H$ is described as follows:
    \[
    \Spec(H) = \left\{ \mu_{\lambda_T}^{\pm} |  \lambda_T \in \Spec (T), \dim \left( \ker (T- \lambda_T I) \cap \mathcal{B}_{\partial X_a}\right) =2 \right\} \cup \{0\}, 
    \]
    where 
    \[
    \mu_{\lambda_T}^{\pm} = \frac{ \mp E_{\lambda_T}(a, a)}{\sqrt{1- \lambda_T^2}}.
    \]
    Moreover, for $\lambda = e^{i\arccos \lambda_T} \in \Spec(U_0)$,  
the vectors
\[
\ket{\psi_\pm^{(\lambda)}}
=
\frac{1}{\sqrt2}\partial_{\lambda}^*f_1^{(\lambda_T)}
\pm
\frac{i}{\sqrt2}\partial_{\lambda}^*f_2^{(\lambda_T)}
\]
are normalized eigenvectors of \(H\) associated with the eigenvalues
$\mu_{\lambda_T}^{\pm},
$ 
respectively.
Furthermore, 
    \[
    \ker H = \mathcal{S}_{\mathrm{sim}} 
    \oplus
    \mathcal{T}_{\text{per}}
    \oplus
    \mathcal{L}^{\perp}.
    \] 
\end{maintheorem2}

Consider the path graph $P_n$, the cycle graph $C_n$, and the complete regular bipartite graph $K_{n,n}$ with an arbitrary reference edge for each case as examples. These three graphs 
are periodic \cite{YNIE, SHH}. The detailed computations are described in Section~5. The dimensions of $\mathcal{S}_{\mathrm{sim}}$,  
    $\mathcal{T}_{\text{per}}$, 
    $\mathcal{L}^{\perp}$, and the robustness $\mathcal{R}_p$ are summarized in Table~\ref{tab:examples} for each case.

In the following, let us observe the robustness of these three examples from the viewpoint of the subspace decomposition of $\ker H = \mathcal{S}_{\mathrm{sim}} 
    \oplus
    \mathcal{T}_{\text{per}}
    \oplus
    \mathcal{L}^{\perp}$. For the transition from $P_N$ to $C_N$ through the connection of both boundary vertices by the edge in $P_N$, the robustness decreases drastically because almost all the eigenvalues in $C_N$ have multiplicity $2$, while all the eigenvalues in $P_N$ are simple, which decreases $\dim \mathcal{S}_{\text{sim}}$ without any increase in $\dim \mathcal{T}_{\text{per}}$. Based on this observation, the existence of cycles seems to decrease the robustness. Then, we consider the complete bipartite graph $K_{n,n}$ $(n\geq 2)$, which is a periodic graph with $K_{1,1}=P_2$, $K_{2,2}=C_4$ and includes many cycles as subgraphs for large $n$.  However, for the transition from $C_n$ to $K_{n,n}$, the robustness {\it increases} due to the growth of both $\dim \mathcal{T}_{\text{per}}$ and $\dim \mathcal{L}^{\perp}$. The growth of $\dim \mathcal{T}_{\text{per}}$ derives from the degeneracy of eigenvalues in $K_{n,n}$; indeed, there are only $3$ distinct eigenvalues. 
The growth of $\dim \mathcal{L}^\perp$ derives from the large number of the first Betti number $|E|-|V|+1$, which is the number of fundamental cycles. 
The dominant contribution comes from $\dim \mathcal{L}^{\perp}$, since $\dim \mathcal{S}_{\mathrm{sim}}$ and $\dim \mathcal{T}_{\text{per}}$ are at most of the order of $|V|$, while $\dim \mathcal{L}^{\perp}$ is of the order of the first Betti number $|E|-|V|+1$. This suggests that graph structures with many edges, for which the contribution of $\mathcal{L}^{\perp}$ becomes large, tend to exhibit stronger robustness of periodicity under magnetic vector potentials. 
However, (interestingly),  the graph having the strongest robustness against the magnetic vector potential is the path that has the smallest number of edges needed to maintain connectivity, that is, \begin{multline*}
0=\lim_{n\to\infty}\mathcal{R}_p(C_{2n})<\cdots<\mathcal{R}_p(C_8)<\mathcal{R}_p(C_6)<
\mathcal{R}_p(C_{4})\\
=\mathcal{R}_p(K_{2,2})<\mathcal{R}_p(K_{3,3})<\cdots<\lim_{n\to\infty}\mathcal{R}_p(K_{n,n})=\mathcal{R}_p(K_{1,1})\\
=\mathcal{R}_p(P_2)=\mathcal{R}_p(P_3)=\cdots=1.
\end{multline*}

\begin{table}[htbp]
 \begin{center}
   \caption{Examples of $\mathcal{R}_p$}
\label{tab:examples}
  \begin{tabular}{|l|c|c|c|c|c|} \hline
     & $|A|$ & $\dim \mathcal{S}_{\text{sim}}$ & $\dim \mathcal{T}_{\text{per}}$ & $\dim \mathcal{L}^{\perp}$ & $\mathcal{R}_p$ \\ \hline
    $P_n$ & $2(n-1)$ & $2(n-1)$ & $0$ & $0$ & $1$ \\ \hline
    $C_{2m}$ & $4m$ & $2$ & $0$ & $2$ & $1/ m$ \\ \hline
    $C_{2m+1}$ & $2(2m +1)$ & $1$ & $0$ & $1$ & $1/ (2m+1)$ \\ \hline
    $K_{n, n}$ & $2n^2$ & $2$ & $4(n-2)$ & $2(n-1)^2$ & $(n^2-2)/n^2$ \\ \hline
  \end{tabular}
 \end{center}
\end{table}
\section{Preliminaries}

\subsection{Graphs}

Let $G=(V, E)$ be a finite, simple, connected, and undirected graph with a vertex set $V$ and an edge set $E$. 
For $x, y \in V$, let $xy$ denote the \emph{directed} edge $(x,y) \in V \times V$ from $x$ to $y$. 
Also, define $A=A(G)$ as the set $\{ xy, yx \mid \{x,y\} \in E\}$. 
Let $a=xy \in A$ be a directed edge.
Let $a^{-1}$ denote the directed edge $yx$.
Let $o(a)$ and  $t(a)$ be the \emph{origin} $x$ and \emph{terminus} $y$ of $a$, respectively. The \emph{degree} of a vertex $x \in V$ for a graph $G$ is written as $\deg x$. The \emph{adjacency matrix} $M=M(G)\in \mathbb{C}^{V\times V}$ of a graph $G$ is defined by
\[M_{x, y}=\begin{cases}
1 & \{x, y\}\in E, \\
0 &\text{otherwise}. 
\end{cases}\]
We define the \emph{degree matrix} $D=D(G) \in \mathbb{C}^{V\times V}$ as $D_{x, y}=(\deg x)\delta_{x, y}$, where $\delta_{x, y}$ is Kronecker's delta. Define the \emph{discriminant matrix} $T=T(G) \in \mathbb{C}^{V\times V}$ by
\[
T=D^{-\frac{1}{2}}MD^{-\frac{1}{2}}.
\]
We note that $T$ is {similar} to the transition matrix of the isotropic random walk on $G$, that is, 
\begin{align}\label{RWbyT}
D^{-\frac{1}{2}}TD^{\frac{1}{2}} =D^{-1}M.
\end{align}

\subsection{Grover walks with a magnetic vector potential }

In this section, we consider a quantum graph walk with a vector potential~\cite{SFRH, QGW2, GK}. Let $G = (V, E)$ be a graph. The vector potential is a one-form, that is,  a function $B : A(G) \to \mathbb{R}$ such that $B(e^{-1}) = -B(e)$. In this paper, we choose the one-form for  a directed edge $a \in A(G)$ and a real value $\beta >0$ so that 
\[
B(e) 
=
\begin{cases}
    \beta & e=a,\\
    -\beta &  e= a^{-1},\\
    0 & \text{otherwise},
\end{cases}
\]
for any $e \in A(G)$. We call the fixed directed edge $a$  the \emph{reference edge}, and $\beta$  the \emph{strength of a vector potential}. A cycle passing through the reference edge acquires a magnetic flux determined by the strength of a vector potential $\beta$~\cite{YT}. 
We explain the relation of our quantum walk to quantum graphs in Section 2.4. 

We now define the quantum walk with a magnetic vector potential considered in this paper.{
\begin{defi}[Grover walk with a magnetic vector potential on $G$]
    Let $G$ be a graph, and let $a \in A$ be a reference edge.
    \begin{itemize}
        \item Total state space: $\mathbb{C}^{A}$ with the standard inner product.
        \item Time evolution: $\psi_t \in \mathbb{C}^{A}$ given by 
        \[
        \psi_t = U_{\beta}\psi_{t-1} \quad (t=1, 2, \dots).
        \] 
        Here, $U_{\beta}=U_{\beta}(G,a) \in \mathbb{C}^{A\times A}$ is defined by 
    \[
    (U_{\beta})_{e,f} 
= e^{iB({e^{-1}})}
\Big(\frac{2}{\deg t(f)} \delta_{o(e), t(f)} - \delta_{e,f^{-1}}\Big),
    \]
    for any $e, f \in A.$
    \item Initial state: $\psi_0 \in \mathcal{H}_0$ with $\|\psi_0 \| = 1$. 
    Here, $\mathcal{H}_0 \subset\mathbb{C}^A$ is defined by \[\mathcal{H}_0 = \left(\,\,
\ker(U_0-I)
\oplus
\ker(U_0+I)
\,\,\right)^\perp\]. 
    \end{itemize}
\end{defi}}
When $\beta=0$, the quantum walk defined by $U_0$ is called the
\emph{Grover walk} on $G$ \cite{SHK}.
Since $U_0$ is unitary, the eigenspaces corresponding to different
eigenvalues are mutually orthogonal. Hence, $\mathcal{H}_0$ is an
invariant subspace of $\mathbb{C}^A$ under $U_0$.
The spectral phases play a central role in characterizing the
periodicity of the walk, as will be discussed in the next subsection.
Moreover, Grover walks always have an eigenvalue $1$ for any connected
graph and have an eigenvalue $-1$ for any connected bipartite graph
\cite{SMT}. These eigenvalues correspond to the spectral phases $0$ and
$\pi$, respectively.
Therefore, we focus on the contribution of eigenvalues other than
$\pm1$ and characterize the robustness of periodicity on the subspace
$\mathcal{H}_0$.

We next introduce several operators.
The following subsection establishes a spectral mapping theorem for the Grover walk using these operators. They also play an important role in the proof of the main result. The \emph{shift matrix} $S_\beta = S_\beta(G,a) \in \mathbb{C}^{A \times A}$
is defined by 
\[
(S_\beta)_{e, f} = \begin{cases}
    e^{iB(f)} & f=e^{-1}, \\
    0 & \text{otherwise}.
\end{cases}
\]  
Clearly,
$S_\beta^2 = I.$
The \emph{boundary matrix} $d = d(G) \in \mathbb{C}^{V \times A}$ is defined by
\[
(d)_{x,e} = \begin{cases}\frac{1}{\sqrt{\deg x}} & x=t(e), \\
0 & \text{otherwise}.
\end{cases}
\]
We also write
\[
C=2d^*d-I,
\]
and refer to $C$ as the \emph{coin matrix}.
For ${\lambda_T \in \Spec(T)}$ and ${\lambda=e^{\pm i\arccos \lambda_T}}$, define the operator\begin{align*}
\partial_\lambda^* =
\begin{cases}
    d^* & \lambda= \pm 1, \\
    {\frac{1}{\sqrt{2(1-\lambda_T^2)}}(I -\lambda S_0)d^*} &  \text{otherwise}.
\end{cases}
\end{align*}
By directly calculating the entries based on the definition of matrix multiplication, it follows that
$dd^* = I$ and $dS_0d^* =T$.
Note that the time evolution matrix $U_{\beta} = U_{\beta}(G,a) \in \mathbb{C}^{A \times A}$ can be written 
as 
\begin{align} \label{Umat}
U_\beta = S_{\beta}C = S_\beta(2d^*d-I).
\end{align} 

\subsection{Spectral properties and periodicity of Grover Walks }

In this section, we review the spectral properties of Grover walks and their periodicity. Periodicity of Grover walks is defined as follows. 
\begin{defi} [\cite{SHK}]
Let $G$ be a graph, and let $U_0=U_0(G)$ be the time evolution matrix of Grover walks. If there exists $\tau \in \mathbb{N}$ such that ${U_0}^\tau=I$, then we say that the graph $G$ is \emph{periodic} and the minimum $\tau$ is \emph{period}. 
Such a graph is also called a $\tau$-periodic graph. 
\end{defi}

\begin{lem}[{\cite{SHH}}] \label{thm:SMT}
Let $G$ be a graph, and let $U_0=U_0(G)$ be the time evolution matrix of Grover walks. 
Then, $G$ is periodic if and only if there exists $\tau \in \mathbb{N}$ such that $\lambda^\tau=1$ for any $\lambda \in \Spec(U_0)$.
\end{lem}

To prepare for the spectral analysis of $U_0$, we introduce the subspace
\[
\mathcal{L}=\operatorname{Im}d^* +\operatorname{Im}S_0d^*.
\]
This subspace is invariant under $U_0$ because of the expression \eqref{Umat}. {By \cite{SMT},} the orthogonal complement of $\mathcal{L}$ can be expressed by 
\begin{align*}
    \mathcal{L}^{\perp}
    &=
    \left(
    \ker d \ \cap \ \ker \left(
    I -S_0
    \right)
    \right)
    \oplus
    \left(
    \ker d \ \cap \ \ker \left(
    I + S_0
    \right)
    \right).
\end{align*}

We next recall the spectral mapping theorem for Grover walks.

\begin{lem}[\cite{SMT}] \label{GSMT}
Let  $G=(V, E)$ be a graph, and let $U_0 = U_0(G)$ be the time evolution matrix of Grover walks. Then, its spectrum is  described as follows:

\[
\Spec(U_0|_{\mathcal{L}})= \left\{e^{\pm i \arccos(\lambda_T)}  | \lambda_T \in \Spec(T) \right\}, 
\]

\[
\Spec(U_0|_{\mathcal{L}^{\perp}})=
\begin{cases}
    \left\{1 \right\} & {\text{$G$ has a single odd cycle,}}\\
    \emptyset  & \text{$G$ is a tree,} \\
    \left\{\pm1 \right\} & \text{otherwise.}
\end{cases}
\]
Moreover, 
\[
\ker (U_0|_{\mathcal{L}} - \lambda I)
=
\left\{
\partial_{\lambda}^* f 
\mid
f \in \ker (T - \lambda_T I)
\right\},
\]

\[
\ker (U_0|_{\mathcal{L}^{\perp}} \mp I)
=
\mathcal{C}_{\pm},
\]
where
\[
\mathcal{C}_{\pm} = \ker d \cap \ker\left(
S {\pm} I
\right),
\]
with 
\[
\begin{cases}
    \dim \mathcal{C}_+ 
    =
    |E| - |V| +1, \\
    \dim \mathcal{C}_- 
    =
    |E| - |V| + \dim \Ker (T + I).
\end{cases}
\]
\end{lem}

As an immediate consequence of the spectral mapping theorem, we obtain the following characterization of the components of eigenvectors of \(U_0\), which will play a key role in the proof of the main result. We also require the following lemma concerning the relationship between the eigenvectors associated with \(\lambda\) and its complex conjugate \(\overline{\lambda}\).

{
\begin{lem} \label{GSMTv}
Let $G$ be a graph, and let $U_0=U_0(G)$. 
For every $\lambda_T \in \Spec(T)\setminus\{\pm1\}$, let  $f_{\lambda_T}\in\ker(T-\lambda_T I)$ and $\lambda=e^{{\pm i\arccos{\lambda_T}}}$.
Then, for every $e\in A(G)$, 
\[
\left(
\partial_\lambda^* f_{\lambda_T}
\right)(e)
= 
\frac{1}{\sqrt{2(1-\lambda_T^2)}}\left(\frac{f_{\lambda_T}({t(e))}}{\sqrt{\deg t(e)} } -\lambda \frac{f_{\lambda_T}({o(e))}}{\sqrt{\deg o(e)} }\right). 
\]
Moreover, if $f_{\lambda_T}$ is chosen to be real-valued, then
\[
    \partial_\lambda^* f_{\lambda_T}(e^{-1}) = 
        -\lambda\overline{\partial_\lambda^* f_{\lambda_T}(e)}. 
    \]
Furthermore, if
$\psi_{\pm}\in\ker(U_0|_{\mathcal L}\mp I)$, then
\[
\psi_\pm(e^{-1})=\pm\psi_\pm(e).
\]
If
$\psi_{\pm}\in\ker(U_0|_{\mathcal L^\perp}\mp I)$, then
\[
\psi_\pm(e^{-1})=\mp\psi_\pm(e).
\]
\end{lem}
}

\begin{lem} \label{GSMTb}
    Let  $G$ be a graph, and let $U_0 = U_0(G)$.
{
For every $\lambda_T \in \Spec(T)\setminus\{\pm1\}$, let  $f_{\lambda_T}\in\ker(T-\lambda_T I)$ and $\lambda=e^{{\pm i\arccos{\lambda_T}}}$. if $f_{\lambda_T}$ is chosen to be real-valued, then
\[
\overline{\partial_\lambda^* f_{\lambda_T}} = \partial_{\overline{\lambda}}^* f_{\lambda_T}.
\]}
\end{lem}

\subsection{Relation to quantum graph with the magnetic vector potential}
We explain why we call $B$ the magnetic vector potential. 
To this end, we review the relation of this quantum walk model to the following quantum graph~{e.g}.\cite{SFRH,BK}, which gives a stationary Schr{\"o}dinger equation on the metric graph. 
The metric graph considered here has the euclidean length $1$ in each arc, and each arc is regarded as the wire with length $1$ and each vertex is the junction. 
Then, the total state space is roughly denoted by $\{\psi\in \mathbb{C}^{A\times [0,1]}\;|\;\psi(e,x)=\psi(\bar{e},1-x) \text{ for any $e\in A$ and $x\in[0,1]$}\}$. 
The pair $(e,x)\in A\times [0,1]$ presents the position with the euclidean distance from $o(e)$ on the arc $e$.  
In the literature on quantum graphs, the magnetic vector potential appears as follows:  
the quantum graph is the stationary Schr{\"o}dinger equation for the plane wave on this metric graph, satisfying the following equation on each wire with the vector potential $B$:
\begin{equation}\label{eq:schro} 
\left(-i\frac{\partial}{\partial x}+B(e)\right)^2\bs{\varphi}(e,x)=k^2 \bs{\varphi}(e,x). 
\end{equation}
A typical boundary condition on each vertex $u\in V$ is given by   
\begin{enumerate}
\item Continuity: 
for $e_1,e_2,\dots,e_s\in A$ with $o(e_j)=u$, 
\[\bs{\varphi}(e_1,0)=\bs{\varphi}(e_2,0)=\cdots=\bs{\varphi}(e_s,0).\] 
\item Current conservation: 
by setting $s_u=\bs{\varphi}(e,0)$ with $o(e)=u$,  
\[ \sum_{o(e)=u}\left(-i\frac{\partial}{\partial x}+B(e)\right)\bs{\varphi}(e,x)\bigg|_{x=0}={-i}\lambda s_u. \]
\end{enumerate}
The parameter $\lambda\in [0,\infty)$ is the height of the delta potential controlling the scattering of the plane wave at each junction, and $k^2$ is the energy of the plane wave. 
A general solution of  (\ref{eq:schro}) with the condition $\bs{\varphi}(e,x)=\bs{\varphi}({e^{-1}},1-x)$ can be described by using two constants $\gamma_e,\gamma_{\bar{e}}\in \mathbb{C}$ such that  
\begin{equation}\label{eq:sol}
\bs{\varphi}(e,x)= \gamma_e\exp[{-i(k+B(e))x}]+\gamma_{\bar{e}}\exp[{-i(k+B(\bar{e}))(1-x)}]. 
\end{equation}
Therefore, the sequence $\{\gamma_e\}_{e\in A}$ isomorphic to the solution of this stationary Schr{\"o}dinger equation. 
By inserting this expression into the boundary conditions (1) and (2), the sequence $\{\gamma_e\}_{e\in A}$ is characterized by $U_\beta$ as follows, which is the relation to our quantum walk: 
\begin{pro}[\cite{QGW1,QGW2,GK}]
Set $\bs{\gamma}\in \mathbb{C}^A$ with $\gamma(e)=\gamma_e$ ($e\in A$) in (\ref{eq:sol}). If $\lambda=0$ and $k=2n\pi$ ($n=1,2,\dots$), then $\bs{\gamma}$ must satisfy
\[ U_\beta \bs{\gamma} = \bs{\gamma}. \]
\end{pro}
A vector $\bs{\gamma}'\in \mathbb{C}^A$ such that there exists $m$ such that $U_\beta^m \bs{\gamma}'=\bs{\gamma}'$ may be considered as an extended above solution. 
In this paper, such a $\bs{\gamma}'$ is arbitrary in $\mathbb{C}^A$ with $m=T$ for $\beta=0$ because the underlying graph $G$ is periodic, and how the perturbation $\beta$ affects it is considered.

\section{Quantum graph walks induced by periodic graphs}
In this section, we present the main theorem describing the effect of a small perturbation induced by a vector potential of the strength $\beta$ on a quantum graph walk on a graph $G$ generating a periodic Grover walk.
{
Note that  we restrict our attention to initial states having no components in the eigenspaces of $U_0$ corresponding to the eigenvalues $\pm1$.
Namely, we consider initial states in
\[
\mathcal{H}_0
:=
\left(
\ker(U_0-I)
\oplus
\ker(U_0+I)
\right)^\perp.
\]
}

We introduce a local decomposition around the reference edge $a \in A(G)$.
\begin{defi}   
    Let $G = (V, E)$ be a graph, and let $a \in A(G)$ be a reference edge. 
We define
\begin{align*}
     \partial X_a &= \{ t(a), o(a) \}, \\
     \partial A_a &= \{ e \in A(G) \mid t(e) \in \partial X_a \}. 
\end{align*}
{
Let us decompose $\mathbb{C}^V$ into the following complementary subspaces such that
\begin{align}\label{Csub}
\mathbb{C}^V=\mathcal{B}_{\partial X_a} \oplus \mathcal{B}_{\partial X_a}'.
\end{align}  Here, $\mathcal{B}_{\partial X_a}'$ is the subspace consisting of functions that vanish on $\partial X_a$, that is   
\[
\mathcal{B}_{\partial X_a}' = \{f \in \mathbb{C}^V \mid \text{for any }x \in \partial X_a , f(x)=0\}.
\]}
\end{defi}
{
Note that, for any $f \in \mathbb{C}^V$, there uniquely exist $f_1 \in \mathcal{B}_{\partial X_a}, f_1' \in \mathcal{B}_{\partial X_a}'$ such that $f=f_1 +f_1'.$ Then, it holds that 
\[
\dim \mathcal{B}_{\partial X_a}' = |V|-2,
\]
which implies 
\[
\dim \mathcal{B}_{\partial X_a}=2.
\]
Therefore, for $\lambda_T \in \Spec(T)$,
\begin{align}\label{atmosttwo}
\dim \left(\ker(T-\lambda_T) \cap \mathcal{B}_{\partial X_a}\right) \leq 2, 
\end{align}
that is, }
 there exist at most two linearly independent eigenvectors in
\[
\ker(T-\lambda_T I) \cap \mathcal{B}_{\partial X_a}.
\]
If $\dim\left(\ker(T-\lambda_T I) \cap \mathcal{B}_{\partial X_a}\right) =2,$ we choose {a} CONS of $\ker (T- \lambda_T I)$ 
as 
\[
f_1^{(\lambda_T)}, f_2^{(\lambda_T)}, \dots, f_{\dim(\ker(T-\lambda_T I))}^{(\lambda_T)} 
\]
so that $f_1^{(\lambda_T)}, f_2^{(\lambda_T)} \in \mathcal{B}_{\partial X_a} $ and $f_j^{(\lambda_T)} \in \mathcal{B}_{\partial X_a}'$ $(j \geq 3)$.  {Since $T$ is a real symmetric matrix, we may choose the above CONS to consist of real vectors.
We remark that}
\begin{align}\label{f}
\ker(T-\lambda_T I) \cap \mathcal{B}_{\partial X_a}
=
{\operatorname{span}}\{ f_1^{(\lambda_T)}, f_2^{(\lambda_T)} \}, 
\end{align}
where $f_1^{(\lambda_T)}$ and $f_2^{(\lambda_T)}$ are real vectors. Concrete forms of \(f_l^{(\lambda_T)}\) are given in Section~5 for the cycle \(C_n\) and the complete bipartite graph \(K_{k,k}\).

\begin{defi}[Hermitian $H$] \label{H1}
    Let $G$ be a graph, and let $a \in A(G)$ be a reference edge. 
    For $e \in A,$ {let} $\iota_{e} : \mathbb{C}^V \to \mathbb{C}^{\{o(e), t(e)\}}$ be defined by 
    \[
    (\iota_{e} f)  = \begin{bmatrix}
        \frac{1}{\sqrt{\deg o(e)}}f(o(e)) \\
        \frac{1}{\sqrt{\deg t(e)}}f(t(e))
    \end{bmatrix}
    \]
    for any $f \in \mathbb{C}^V.$
    For $\lambda_T \in \Spec(T) $ with $\ker(T- \lambda_T I) \cap \mathcal{B}_{\partial X_a} = \mathrm{span} \{ f_1 ^{(\lambda_T)}, f_2 ^{(\lambda_T)}\}$ we set the area spanned by $\iota_{e_1}f_1^{(\lambda_T)}, \iota_{e_2}f_2^{(\lambda_T)}$ by 
    \[
    E_{\lambda_T} (e_1, e_2) = \det \begin{bmatrix}
        \iota_{e_1}f_1^{(\lambda_T)} & \iota_{e_2}f_2^{(\lambda_T)}
    \end{bmatrix}
    \]
    for $e_1, e_2 \in A.$
    For $\lambda_T \in \Spec(T)$ with $\dim(\ker(T- \lambda_T I)\cap \mathcal{B}_{\partial X_a}) \leq 1, $
    we set $E_{\lambda_T}(e_1, e_2)= 0.$
    Then, the Hermitian $H = H(G, a) \in \mathbb{C}^{A \times A}$ is defined by 
    \[
    (H)_{e_1, e_2} = 
    \sum_{\lambda_T \in \Spec(T) \setminus\{\pm1\}}\frac{1}{1-\lambda_T^2}E_{\lambda_T} (a, a)
    \left(  
    E_{\lambda_T} (e_1, e_2)
    +
    E_{\lambda_T} (e_2, e_1)
    \right).
    \]
\end{defi}
{
Let us see that the Hermitian matrix $H$ is independent of choice of a CONS consisting of real vectors. Assume \[\dim\left(\ker(T-\lambda_T I)\cap \mathcal{B}_{\partial X_a}\right)=2,\] and put \[\dim \ker(T-\lambda_T I)=s \] for $s\geq 2$. Then, set 
\[
\ker(T-\lambda_T I) = \operatorname{span}\{f_1^{(\lambda_T)}, f_2^{(\lambda_T)}, \dots, f_s^{(\lambda_T)}\} 
\] 
with \begin{align}\label{kitei1}
        f_1^{(\lambda_T)}, f_2^{(\lambda_T)} \in \mathcal{B}_{\partial X_a}, \qquad
        f_j^{(\lambda_T)} \in \mathcal{B}_{\partial X_a}' \quad(j >2). 
\end{align}
On the other hand, set the another CONS of $\ker(T-\lambda_T I)$ by 
\[
\{\widetilde f_1^{(\lambda_T)}, \widetilde f_2^{(\lambda_T)}, \dots, \widetilde f_s^{(\lambda_T)}\}
\]
with \begin{align}\label{kitei2}
        \widetilde f_1^{(\lambda_T)}, \widetilde f_2^{(\lambda_T)} \in \mathcal{B}_{\partial X_a}, \qquad
        \widetilde f_j^{(\lambda_T)} \in \mathcal{B}_{\partial X_a}' \quad(j >2). 
\end{align}
Then, there exists an $s$-dimensional orthogonal matrix $Q$ such that 
\[
\begin{bmatrix}
    \widetilde f_1^{(\lambda_T)} 
    \widetilde f_2^{(\lambda_T)}\cdots \widetilde f_s^{(\lambda_T)}
\end{bmatrix}
=
\begin{bmatrix}
     f_1^{(\lambda_T)}  f_2^{(\lambda_T)}\cdots f_s^{(\lambda_T)}
\end{bmatrix}
Q.
\]
We note that by the restriction of \eqref{kitei1} and \eqref{kitei2}, $Q$ can be decomposed to $2$-dimensional and $(s-2)$-dimensional orthogonal matrices $Q_2$ and $Q'$, respectively, such that 
\[
Q = \begin{bmatrix}
    Q_2 &O \\
    O& Q'
\end{bmatrix}.
\]
Then, we have
\[
\begin{bmatrix}
\iota_{e_1}
\widetilde f_1^{(\lambda_T)} &
\iota_{e_2}
\widetilde f_2^{(\lambda_T)}
\end{bmatrix}
=
\begin{bmatrix}
\iota_{e_1}
f_1^{(\lambda_T)} &
\iota_{e_2}
f_2^{(\lambda_T)}
\end{bmatrix}Q_2
\]
for any $e_1, e_2 \in A,$
which implies
\[
\widetilde E_{\lambda_T}(a,a)
=
\det(Q_2)E_{\lambda_T}(a,a),
\]
and 
\[
\widetilde E_{\lambda_T}(e_1,e_2)+\widetilde E_{\lambda_T}(e_2,e_1)
=
\det(Q_2)\left(E_{\lambda_T}(e_1,e_2)+E_{\lambda_T}(e_2,e_1)\right).
\]
Hence the factor
\[
E_{\lambda_T}(a,a)
\left(
E_{\lambda_T}(e_1,e_2)+E_{\lambda_T}(e_2,e_1)
\right)
\]
is invariant under the change of basis, since $\det(Q_2)^2=1$.
Consequently, $H$ is independent of the choice of the real orthonormal basis.}

\begin{defi}\label{ST}
    Let $G$ be a graph, and let $U_0=U_0(G)$ be the time evolution matrix of Grover walks. 
    If $\dim\left(\ker(T-\lambda_T I)\cap \mathcal{B}_{\partial X_a}\right) = 1$, we say that $\lambda_T$ has a simple eigenvector in $\mathcal{B}_{\partial X_a}$ and denote such an eigenvector by $g_{\lambda_T}$.
    Define 
    \[
    \mathcal{S}_{\mathrm{sim}}
    =
    \bigoplus_{\substack{
    \lambda_T\in\Spec(T) \\ 
    \lambda_T \text{  has a simple eigenvector in $\mathcal{B}_{\partial X_a}$}}}{\operatorname{span}}
    \left\{
    \partial_{\lambda}^*
    g_{\lambda_T} \mid  \lambda=e^{\pm i \arccos \lambda_T}
    \right\},
    \]
    and 
    \[
\mathcal{T}_{\mathrm{per}}
=
\bigoplus_{\substack{
\lambda_T\in\Spec(T)
}}{\operatorname{span}}
\left\{
\partial_{\lambda}^*
f \mid 
f \in 
\ker(\lambda_T I-T)\cap \mathcal{B}_{\partial X_a}' ,\  \lambda= e^{\pm i \arccos \lambda_T}
\right\}.
    \]
\end{defi}

In Section~5, we discuss several explicit examples of the Hermitian matrix $H$ for particular graphs.

Using the Hermitian matrix \(H\) defined above, we now state the main theorem of this paper. The first main result shows that a periodic Grover walk with a  small perturbation induced by a vector potential admits an effective continuous-time description.

\begin{thm}\label{main}
    Let $G$ be a graph that induces a $\tau$-periodic Grover walk. For the small strength {$\beta$} of a vector potential and a reference edge $a \in A(G)$, denote $U_{\beta} = U_{\beta}(G,a)$ and $H = H(G,a)$.   Set $\phi_0  {\in \mathcal{H}_0}$ as the initial state, and $\varphi_{\lceil \frac{t}{\beta}\rceil}^{[D]}$   as the state at the final time 
    $\lceil \frac{t}{\beta}\rceil$ of the discrete-time quantum walk with some $t>0$, defined by the following time evolution:
\[
\varphi_0^{[D]} = \phi_0;\  \varphi_n^{[D]} = U^{\tau}_{\beta} \varphi_{n-1}^{[D]} \ 
\left(n=1,\dots, \left\lceil \frac{t}{\beta} \right \rceil
\right).
\]
     Set $\psi_t^{[C]}$ as the state at time 
    $t$ of the continuous-time quantum walk, defined by the following time evolution:
    {
    \begin{align*}
\psi_0^{[C]} = \phi_0;\  
i \frac{\partial}{\partial t}\psi_t^{[C]} = \tau H\psi_t^{[C]} \ (t>0).
\end{align*}
}
Then, we have
\[
\psi_t^{[C]} = \lim_{\beta \to 0} \varphi_{\lceil \frac{t}{\beta}\rceil}^{[D]}.
\]

\end{thm}

The second main result concerns robustness of periodicity under a  small perturbation induced by a vector potential.

\begin{thm}\label{main2}
    Set $U_{0} = U_{0}(G)$ and $H = H(G, a)$, for a graph $G$ that induces $\tau$-periodic Grover walks and a reference edge $a \in A(G)$. For  $\lambda = e^{\pm i \arccos{\lambda_T}} \in \Spec(U_0)$ with $\dim \left( \ker (T- \lambda_T I) \cap \mathcal{B}_{\partial X_a}\right) =2$, we denote by $\ket{\lambda^{(1)}}, \ket{\lambda^{(2)}}$ an orthonormal basis of $\ker(U_0-\lambda I).$ Then, the spectrum of $H$ is described as follows:
    \[
    \Spec(H) = \left\{ \mu_{\lambda_T}^{\pm} |  \lambda_T \in \Spec (T), \dim \left( \ker (T- \lambda_T I) \cap \mathcal{B}_{\partial X_a}\right) =2 \right\} \cup \{0\}, 
    \]
    where 
    \begin{align}\label{mu}
    \mu_{\lambda_T}^{\pm} = \frac{ \mp E_{\lambda_T}(a, a)}{\sqrt{1- \lambda_T^2}}.
    \end{align}
Moreover, for $\lambda = e^{i\arccos \lambda_T} \in \Spec(U_0)$,   
the vectors
\begin{align}\label{muv}
\ket{\psi_\pm^{(\lambda)}}
=
\frac{1}{\sqrt2}\partial_{\lambda}^*f_1^{(\lambda_T)}
\pm
\frac{i}{\sqrt2}\partial_{\lambda}^*f_2^{(\lambda_T)}
\end{align}
are normalized eigenvectors of \(H\) associated with the eigenvalues
$\mu_{\lambda_T}^{\pm},
$ 
respectively.
Furthermore, \begin{align}\label{kerh}
    \ker H = \mathcal{S}_{\mathrm{sim}} 
    \oplus
    \mathcal{T}_{\text{per}}
    \oplus
    \mathcal{L}^{\perp}.
    \end{align}
\end{thm}

\section{Proof of main theorem}

In this section, we present several definitions and lemmas for the proof of Theorem~\ref{main} and Theorem~\ref{main2}, and conclude with the proof at the end of the section.

\begin{defi}[Projection of $\partial X_a$ and $\partial A_a$]
    Let $G$ be a graph, and let $a \in A(G)$ be a reference edge.
    We define $\Pi_{\partial X_a}$ and $\Pi_{\partial A_a}$ as the projections onto the subspaces corresponding to the vertices in $\partial X_a$ and the directed edges  in $\partial A_a$, respectively. That is, 
\begin{align*}
    (\Pi_{\partial X_a})_{x,y} = 
    \begin{cases}
        1 & x=y, \ x\in \partial X_a, \\
        0 & \text{otherwise},
    \end{cases}
\end{align*}
 and
 \begin{align*}
    (\Pi_{\partial A_a})_{e,f} = 
    \begin{cases}
        1 & e=f, \ e\in \partial A_a, \\
        0 & \text{otherwise}.
    \end{cases}
\end{align*}
Also, let $\sigma = \sigma(G, a) \in \mathbb{C}^{A(G)\times A(G)}$ be the matrix defined by
\[
(\sigma)_{e, f}=\begin{cases}
    1 & e=f=a, \\
    -1 & e=f=a^{-1}, \\
    0 & \text{otherwise}.
\end{cases}
\]
\end{defi}

\begin{lem}\label{commute}
    Let $G$ be a graph. For a strength {$\beta$} of a vector potential and a reference edge $a \in A(G)$, set $U_{\beta} = U_{\beta}(G,a) = S_{\beta}C$. 
    Then $C$ and $\Pi_{\partial A_a}$ commute. Moreover, $\sigma S_{\beta}$ and $\Pi_{\partial A_a}$ also commute.
\end{lem}
\begin{proof}
    This can be checked by a direct computation of the matrix elements.
\end{proof}

In this section, for each
\(\lambda\in\Spec(U_0)\),
we denote by
\[
\{\ket{\lambda^{(1)}},\dots,\ket{\lambda^{(r_\lambda)}}\}
\]
an orthonormal basis of
\[
\ker(\lambda I-U_0),
\]
where
\[
r_\lambda=\dim\ker(\lambda I-U_0).
\]
Furthermore, for each basis vector $\ket{\lambda^{(j)}}$, we denote by $(\ket{\lambda^{(j)}})_e$ its component indexed by $e\in A$.

\begin{lem} \label{braket}
    Let $\ket{\lambda}$ be an eigenvector of a time evolution matrix $U_{0}$ associated with the eigenvalue $\lambda$. Then the following holds.
    \[
    \bra{\lambda}\sigma\ket{\lambda}=0.
    \]
\end{lem}

\begin{proof}
    By the definition of $\sigma$, $\bra{\lambda}\sigma\ket{\lambda}$ 
    is determined by the components corresponding to $a^{\pm}$. Thus,
    \begin{align*}
    \bra{\lambda}\sigma\ket{\lambda}
    &=
    \begin{bmatrix}
        \overline{(\ket{\lambda})_{a}} & \overline{(\ket{\lambda})_{a^{-1}}}
    \end{bmatrix}
    \begin{bmatrix}
        1 & 0 \\
        0 & -1
    \end{bmatrix}
    \begin{bmatrix}
        (\ket{\lambda})_{a} \\ (\ket{\lambda})_{a^{-1}}
    \end{bmatrix} \\
    &=
    |(\ket{\lambda})_{a}|^2 - |(\ket{\lambda})_{a^{-1}}|^2.
\end{align*}
Since $U_0$ is unitary, we have $|\lambda|=1$. 
Therefore, by Lemma~\ref{GSMTv}, $|(\ket{\lambda})_{a}| = |(\ket{\lambda})_{a^{-1}}|$, and the lemma follows.
\end{proof}

\begin{lem}\label{braket1}
Let $\lambda \in \{\pm1\}$ be an eigenvalue of {$U_0|_{\mathcal{L}^{\perp}}$} with multiplicity
$m>1$, and let
$\ket{\lambda^{(j)}}$ $(1 \leq j \leq m)$
denote the corresponding eigenvectors. 
Then we have
\[
\bra{\lambda^{(j)}}\sigma\ket{\lambda^{(k)}}=0
\]
for $j \neq k.$
\end{lem}

\begin{proof}
    The proof is essentially the same as that of Lemma~\ref{braket}.
\end{proof}

We note that, for $\lambda \neq \pm 1$, the quantity
\[
\bra{\lambda^{(j)}}\sigma\ket{\lambda^{(k)}}
\]
does not necessarily vanish when $j \neq k.$
Motivated by this observation, we introduce the following matrix whose entries are given by these quantities.

\begin{defi}
    Let $\lambda_j$ be an eigenvalue of $U_{0}$ with multiplicity $r_j$, and let $\{\ket{\lambda_j^{(l)}}\}_{l=1}^{r_j}$ be an orthonormal basis of its eigenspace corresponding to $\lambda_j$. We define 
    $\Phi_{\lambda_j}$,
    and $M_{\sigma} ^{(\lambda_j, \lambda_k)}$ as follows:
    \begin{align*}
    &\Phi_{\lambda_j} :=\left[\ \ket{\lambda_j^{(1)}} \cdots \ket{\lambda_j^{(r_j)}} \ \right], \\
    &(M_{\sigma} ^{(\lambda_j, \lambda_k)})_{l, l'}
    := \bra{\lambda_j^{(l)}}  \sigma \ket{\lambda_k^{(l')} },
    \end{align*}
    for $1\le l\le r_j$ and $1\le l'\le r_k$.
\end{defi}

Here, for a matrix $F = \mathcal{O}(\beta^m) \ (m \geq 1)$ means that 
\[
0 \leq \limsup_{\beta \to 0 }{ \beta^{-m}\max_{l, l' \in A(G)}\left|(F)_{l, l'}\right|} < \infty.
\]
We are now ready to state a key lemma of this paper.

\begin{lem} \label{key lemma}
    Let $G$ be a graph that induces $\tau$-periodic Grover walks.  For a strength {$\beta$} of a vector potential  and a reference edge $a \in A(G)$, set $U_{\beta} = U_{\beta}(G,a)$. Here, we assume that $0 < \beta \ll 1 $. Let $H = H(G, a)$ defined by Definition~ \ref{H1}. Then, the following holds for $U_{\beta}^{\tau}$.
    \begin{align*}
       U_{\beta}^{\tau}
       =
       I {-} 
       i\beta \tau \left(H{+\sum_{\lambda\in {\{\pm1\}}}
    \Phi_{\lambda} M_{\sigma}^{(\lambda, \lambda)} \Phi_{\lambda} ^*}\right) +
        \mathcal{O}(\beta^2).
    \end{align*}
\end{lem}
\begin{proof}
    By setting $U_0=S_0C$ and $U'\in \mathbb{C}^{A(G)\times A(G)}$ such that
    \[
    U'= (S_{\beta}-S_0)C\Pi_{\partial A},
    \]
    then, $U_{\beta}=U_0+U'$.
    Note that, 
    \[
    (S_{\beta}-S_0)_{e, f}=
    \begin{cases}
    e^{i\beta}-1 & e=a^{-1}, f=a, \\
    e^{-i\beta}-1 & e=a, f=a^{-1}, \\
    0 & \text{otherwise}.
    \end{cases}
    \]
    Since $\beta$ is sufficiently small, the right-hand side can be expanded in powers of $\beta$ up to second order, we obtain
    \begin{align*}
        (S_{\beta}-S_0)_{e,f} = \begin{cases}
            i\beta + \mathcal{O}(\beta^2) & e=a^{-1}, f = a,  \\
            -i\beta + \mathcal{O}(\beta^2) & e=a,  f = a^{-1},  \\
            0 & \text{otherwise}.
        \end{cases}
    \end{align*}
    Therefore, $S_{\beta}-S_0=({-}i\beta \sigma + \mathcal{O}(\beta^2))S_0$. Then, 
    \begin{align*}
        U'&= {-}i\beta \sigma S_0C\Pi_{\partial A_a} + \mathcal{O}(\beta^2)  \\
        &= {-}i\beta \Pi_{\partial A_a} \sigma S_0C + \mathcal{O}(\beta^2)   \\
        &= {-}i\beta  \sigma U_0 + \mathcal{O}(\beta^2). 
    \end{align*}
    Here, the second equality follows from Lemma~\ref{commute}, and the third equality follows from the fact that $\Pi_{\partial A}\sigma=\sigma$. 
    Since $G$ is $\tau$-periodic, taking the $\tau$-th power of $U_{\beta}$, we obtain 
    \begin{align}
        U_{\beta}^{\tau} &=  ( U_0 {-} i\beta  \sigma U_0 + \mathcal{O}(\beta^2) )^{\tau} \notag \\
        &= I {-} i\beta \sum_{j=0}^{\tau-1}U_0^j \sigma U_0^{-j}  +\mathcal{O}(\beta^2). \label{tenkai}
    \end{align}
    Using the spectral decomposition, {we} evaluate the first-order term in $\beta$ separately for each eigenvalue.
    Let $P_{\lambda}$ denote the spectral projection corresponding to the eigenvalue $\lambda$. Then, the first-order term in $\beta$ on RHS in the above equality can be decomposed as
    \begin{align} \sum_{j=0}^{\tau-1} U_0^j \sigma U_0^{-j} 
    &= \sum_{j=0}^{\tau-1}\bigg(\sum_{\lambda \in \Spec (U_0)}\lambda^j P_{\lambda} \bigg)
    \sigma 
    \bigg(\sum_{\lambda' \in \Spec (U_0)}\lambda'^{-j}P_{\lambda'} \bigg) \notag \\
    &=
    \sum_{\lambda, \lambda' \in \Spec (U_0)}
    \sum_{j=0}^{\tau-1}
    \left(\frac{\lambda}{\lambda'}\right)^j P_{\lambda}\sigma P_{\lambda'} .\label{eq;first}
    \end{align}
    Since $G$ is $\tau$-periodic, $\left( \lambda / \lambda' \right)^{\tau}=1$ from Lemma~\ref{thm:SMT}. 
    Then, the term with $\lambda \neq \lambda'$ in \eqref{eq;first} has no contribution because  
    \begin{align*}
     \sum_{j=0}^{\tau-1}
    \left(\frac{\lambda}{\lambda'}\right)^j P_{\lambda}\sigma P_{\lambda'} 
    &=
    \frac{1-\left(\frac{\lambda}{\lambda'}\right)^{\tau}}{1-\left(\frac{\lambda}{\lambda'}\right)}P_{\lambda}\sigma P_{\lambda'} =0.
    \end{align*}
    Therefore, it {is} sufficient to consider only the terms with $\lambda=\lambda'$.
    \begin{align*}
    \sum_{\lambda, \lambda' \in \Spec (U_0)}
    \sum_{j=0}^{\tau-1}
    \left(\frac{\lambda}{\lambda'}\right)^j P_{\lambda}\sigma P_{\lambda'} 
    &= 
     \sum_{\lambda \in \Spec (U_0)}
    \sum_{j=0}^{\tau-1}
    P_{\lambda}\sigma P_{\lambda} \\
    &= 
    \sum_{\lambda \in \Spec (U_0)}
    \tau
    P_{\lambda}\sigma P_{\lambda} \\
    &=
    \sum_{\lambda\in \Spec(U_0)} 
    \tau
    \Phi_{\lambda} M_{\sigma}^{(\lambda, \lambda)} \Phi_{\lambda} ^*.
    \end{align*}
   {Since the terms corresponding to $\lambda=\pm1$ already have the desired form, it remains to consider the case $\lambda\neq\pm1$. }
    By Lemma~\ref{GSMTv}, we express $M_{\sigma}^{(\lambda, \lambda)}$ for  $\lambda \neq \pm1$ below. Here, set $\lambda_T = \Re(\lambda)$ for $\lambda \in \Spec(U_0|_{\mathcal{L}}).$ Let $f^{(\lambda_T)}_{j}$ be a real eigenvector of $T$ corresponding to the eigenvalue $\lambda_T$ associated with $\ket{\lambda^{(j)}}$. 
    {
\begin{align}
    \left(M_{\sigma}^{(\lambda, \lambda)}\right)_{l, l'} 
    &=
    \bra{\lambda^{(l)}}\sigma\ket{\lambda^{(l')}} \notag \\
    &=
    \begin{bmatrix}
        \overline{(\ket{\lambda^{(l)}})_{a}} & \overline{(\ket{\lambda^{(l)}})_{a^{-1}}}
    \end{bmatrix}
    \begin{bmatrix}
        1 & 0 \\
        0 & -1
    \end{bmatrix}
    \begin{bmatrix}
        (\ket{\lambda^{(l')}})_{a} \\ (\ket{\lambda^{(l')}})_{a^{-1}}
    \end{bmatrix} \notag \\ 
    &=
    \overline{(\ket{\lambda^{(l)}})_{a}}
    (\ket{\lambda^{(l')}})_{a}
    -
    \overline{(\ket{\lambda^{(l)}})_{a^{-1}}}
    (\ket{\lambda^{(l')}})_{a^{-1}} \notag \\
    &=
    \overline{(\ket{\lambda^{(l)}})_{a}}
    (\ket{\lambda^{(l')}})_{a}
    -
    \overline{\overline{(\ket{\lambda^{(l)}})_{a}}
    (\ket{\lambda^{(l')}})_{a} }\notag \\
    &=
    2i\Im{\left(\overline{(\ket{\lambda^{(l)}})_{a}}
    (\ket{\lambda^{(l')}})_{a}\right)} \notag \\
    &=
    \frac{i\Im\left(
    \left(\frac{f_{l}^{(\lambda_T)}({t(a))}}{\sqrt{\deg t(a)} } -\overline{\lambda} \frac{f_{l}^{(\lambda_T)}({o(a))}}{\sqrt{\deg o(a)}} \right)
    \left(\frac{f_{l'}^{(\lambda_T)}({t(a))}}{\sqrt{\deg t(a)} } -\lambda \frac{f_{l'}^{(\lambda_T)}({o(a))}}{\sqrt{\deg o(a)}} \right)\right)}{1-\lambda_T^2} \notag\\
    &=
    \frac{i\left(
    \Im\left(-\lambda \frac{f_{l'}^{(\lambda_T)}({o(a))}}{\sqrt{\deg o(a)}}\frac{f_{l}^{(\lambda_T)}({t(a))}}{\sqrt{\deg t(a)} }\right)  -\Im\left(\overline{\lambda} \frac{f_{l}^{(\lambda_T)}({o(a))}}{\sqrt{\deg o(a)}} \frac{f_{l'}^{(\lambda_T)}({t(a))}}{\sqrt{\deg t(a)} }
    \right)\right) }{1-\lambda_T^2} \notag\\
    &=
    \frac{i\left(
    -f_{l'}^{(\lambda_T)}({o(a))}f_{l}^{(\lambda_T)}({t(a))}
    +
    f_{l}^{(\lambda_T)}({o(a))}f_{l'}^{(\lambda_T)}({t(a))}
    \right)}{\sqrt{ (1- {\lambda_T}^2) \deg o(a) \deg t(a)} }
     \notag \\
    &=
    \frac{i}{\sqrt{ (1- {\lambda_T}^2) \deg o(a) \deg t(a)} }
    \begin{vmatrix}f_l^{(\lambda_T)} (o(a)) &
    f_{l'}^{(\lambda_T)} (o(a)) \\
    f_l^{(\lambda_T)} (t(a)) &
    f_{l'}^{(\lambda_T)} (t(a))\end{vmatrix}, \label{|f|}
\end{align}
where the fourth equality follows from Lemma~\ref{GSMTv} and $\lambda\overline{\lambda} =|\lambda|^2=1$,  the sixth equality follows from Lemma~\ref{GSMTv}, and the eighth equality follows from $\Im(\lambda)=\sqrt{1-\lambda_T^2}$.}
It is easy to see that 
\[
\left(M_\sigma ^{(\lambda, \lambda)}\right)_{l, l'}
=
-\left(M_\sigma ^{(\lambda, \lambda)}\right)_{l', l}
\]
We note that there exist \emph{at most} two linearly independent vectors in { $\ker(T-\lambda_T I)\cap \mathcal{B}_{\partial X_a}$ by \eqref{atmosttwo} }. 
Then, we may choose an orthogonal real basis of $\ker(T-\lambda_T I)$ such that at most two basis vectors have an intersection with $\partial X_a$. If they exist, we denote them by $f_1^{(\lambda_T)}$ and $f_2^{(\lambda_T)}$, that is, 
\[
\braket{f_l^{(\lambda_T)}, f_{l'}^{(\lambda_T)} } = \delta_{l, l'},
\]
and {\[
\{x \in \partial X_a \mid f_l^{(\lambda_T)}(x) \neq 0\}\neq \emptyset,
\]} for $l, l' \in \{ 1,2 \}.$
Therefore, under this choice of basis vectors, we obtain
\begin{itemize}
    \item for $\lambda$ with  $\dim \left( \ker (T- \lambda_T I) \cap \mathcal{B}_{\partial X_a}\right) \leq 1$,
\end{itemize}
\begin{align*}
    M_{\sigma}^{(\lambda, \lambda)} = O.
\end{align*}
\begin{itemize}
    \item for $\lambda$ with  $\dim \left( \ker (T- \lambda_T I) \cap \mathcal{B}_{\partial X_a}\right) =2$,
\end{itemize}
\begin{align*}
    M_{\sigma}^{(\lambda, \lambda)}
    &=
    \begin{bmatrix}
    \bra{\lambda^{(1)}}\sigma\ket{\lambda^{(1)}} & \bra{\lambda^{(1)}}\sigma\ket{\lambda^{(2)}} \\
    \bra{\lambda^{(2)}}\sigma\ket{\lambda^{(1)}} & \bra{\lambda^{(2)}}\sigma\ket{\lambda^{(2)}}
    \end{bmatrix}
    \oplus O \\
    &=
    \frac{i E_{\lambda_T}(a, a)}{\sqrt{1 - \lambda_T^2}}
    \begin{bmatrix}
    0 & 1 \\
    -1 & 0
    \end{bmatrix}
    \oplus O. 
\end{align*}
Here, the second equality is derived from Lemma~\ref{braket} and \eqref{|f|}. 
Then, for $\lambda$ with  $\dim \left( \ker (T- \lambda_T I) \cap \mathcal{B}_{\partial X_a}\right) =2$,

\begin{align}
    \Phi_{\lambda} M_{\sigma}^{(\lambda, \lambda)} \Phi_{\lambda} ^*
    &=
    \frac{i E_{\lambda_T}(a, a)}{\sqrt{1- \lambda_T^2}}
    \begin{bmatrix}
    \ket{\lambda^{(1)}} &\ket{\lambda^{(2)}}
    \end{bmatrix}
    \begin{bmatrix}
        0 & 1 \\
        -1 & 0
    \end{bmatrix}
    \begin{bmatrix}
    \bra{\lambda^{(1)}} \\
    \bra{\lambda^{(2)}}
    \end{bmatrix} \label{SD} \\
    &=
    \frac{i E_{\lambda_T}(a, a)}{\sqrt{1- \lambda_T^2}}
    \left(
    \ket{\lambda^{(1)}}\bra{\lambda^{(2)}}
    -
    \ket{\lambda^{(2)}}\bra{\lambda^{(1)}}
    \right), \label{H-braket}  
\end{align}
and hence
\begin{align*}
     \overline{\Phi_{\lambda} M_{\sigma}^{(\lambda, \lambda)} \Phi_{\lambda} ^*}
     &=
     -\frac{i E_{\lambda_T}(a, a)}{\sqrt{1- \lambda_T^2}}
    \left(
    \ket{\overline{\lambda}^{(1)}}\bra{\overline{\lambda}^{(2)}}
    -
    \ket{\overline{\lambda}^{(2)}}\bra{\overline{\lambda}^{(1)}}
    \right) \\
    &=
    \Phi_{\overline{\lambda}} M_{\sigma}^{(\overline{\lambda}, \overline{\lambda})} \Phi_{\overline{\lambda}} ^*.
\end{align*}
Here, the first equality follows from Lemma~\ref{GSMTb}, and the second equality follows from $\Im{(\overline{\lambda})} = -\sqrt{1- \lambda_T^2}$.
   Thus, 
   \begin{align*} 
       &(\Phi_{\lambda} M_{\sigma}^{(\lambda, \lambda)} \Phi_{\lambda} ^*
       +
       \Phi_{\overline{\lambda}} M_{\sigma}^{(\overline{\lambda}, \overline{\lambda})} \Phi_{\overline{\lambda}} ^*)_{e_1, e_2} \\
       &=
       2\Re(\Phi_{\lambda} M_{\sigma}^{(\lambda, \lambda)} \Phi_{\lambda} ^*)_{e_1, e_2}
       \\
       &=
       \frac{2 E_{\lambda_T}(a, a)}{\sqrt{1- \lambda_T^2}}
       \Re\left(
       i\ket{\lambda^{(1)}}\bra{\lambda^{(2)}}
       -
       i\ket{\lambda^{(2)}}\bra{\lambda^{(1)}}
       \right)_{e_1, e_2}
       \\
       &=
       \frac{-2 E_{\lambda_T}(a, a)}{\sqrt{1- \lambda_T^2}}
       \Im\left(
       \ket{\lambda^{(1)}}_{e_1}\bra{\lambda^{(2)}}_{e_2}
       -
       \ket{\lambda^{(2)}}_{e_1}\bra{\lambda^{(1)}}_{e_2}
       \right)
       \\
       &=
       \frac{- E_{\lambda_T}(a, a)}{\sqrt{1- \lambda_T^2}(1-\lambda_T^2)}
       \Im
       \left(
       -E_{\lambda_T}(e_1, e_2) \lambda
       +
       E_{\lambda_T}(e_2, e_1)\overline{\lambda}
       \right)
       \\
       &=
       \frac{- E_{\lambda_T}(a, a)}{\sqrt{1- \lambda_T^2}(1-\lambda_T^2)}
       \left(
       -E_{\lambda_T}(e_1, e_2)\Im(\lambda)
       +
       E_{\lambda_T}(e_2, e_1)\Im(\overline{\lambda})
       \right)
       \\
       &=
       \frac{E_{\lambda_T}(a, a)}{1-\lambda_T^2}
       \left(
       E_{\lambda_T}(e_1, e_2)
       +
       E_{\lambda_T}(e_2, e_1)
       \right),
   \end{align*}
   where the forth equality follows from  Lemma~\ref{GSMTv}, and the final equality follows from $\Im (\lambda) = -\Im (\overline{\lambda})=\sqrt{1-\lambda_T^2}$.
       Therefore, the second term in \eqref{tenkai} can be written as
       \begin{align}
           \left(\sum_{j=0}^{\tau-1}U_0^j \sigma U_0^{-j}\right)_{e_1,e_2}
           &=
           \sum_{\substack{
           \lambda \in \Spec(U_0) \setminus\{\pm1\}\\
           \dim \left( \ker (T- \lambda_T I) \cap \mathcal{B}_{\partial X_a}\right)=2 }}
           \tau\left(
           \Phi_{\lambda} M_{\sigma}^{(\lambda, \lambda)} \Phi_{\lambda} ^*\right)_{e_1, e_2} \notag\\
           &=
           \sum_{\substack{
           \lambda_T \in \Spec(T) \setminus\{\pm1\}\\
           \dim \left( \ker (T- \lambda_T I) \cap \mathcal{B}_{\partial X_a}\right)=2}}
           \tau
           \left(
           \Phi_{\lambda} M_{\sigma}^{(\lambda, \lambda)} \Phi_{\lambda} ^*
           +
           \Phi_{\overline{\lambda}} M_{\sigma}^{(\overline{\lambda}, \overline{\lambda})} \Phi_{\overline{\lambda}} ^* 
           \right)_{e_1, e_2} \notag\\
           &=
           \sum_{\lambda_T\in \Spec(T)\setminus \{\pm 1\}}
           \tau\frac{E_{\lambda_T}(a, a)}{1-\lambda_T^2}
           \left(
           E_{\lambda_T}(e_1, e_2)
           +
           E_{\lambda_T}(e_2, e_1)
           \right) \notag\\
           &=
           \tau(H)_{e_1, e_2}. \label{H1 2} 
       \end{align}
\end{proof}

{Throughout this section, for a matrix $\Lambda \in \mathbb{C}^{A\times A}$, $\|\Lambda\|$ denotes the operator norm induced by the Euclidean norm, that is, \[ \|\Lambda\|=\sup_{\|x\|=1}\|\Lambda x\|. \]}

\begin{cor} \label{beta^2}
Let $G$ be a graph that induces $\tau$-periodic Grover walks.
For a strength $\beta$ of a vector potential and a reference edge
$a \in A(G)$, set $U_{\beta}=U_{\beta}(G,a)$.
Here, we assume that $0<\beta\ll1$.
Let $H=H(G,a)$ be defined by Definition~\ref{H1}.
Then, the following holds:
\[    \left\|
    \left(
    e^{-i\beta\tau H}-U_{\beta}^{\tau}
    \right){|_{\mathcal{H}_0}}
    \right\|
    =
    \mathcal{O}(\beta^2).
   \]
\end{cor}


\begin{proof}
{
By the Taylor expansion of the finite and Hermitian matrix exponential,
\[
e^{-i\beta\tau H}
=
I-i\beta\tau H+\mathcal{O}(\beta^2).
\]
On the other hand, Lemma~\ref{key lemma} gives
\[
U_\beta^\tau
=
I-i\beta\tau \left(H+\sum_{\lambda\in {\{\pm1\}}}
    \Phi_{\lambda} M_{\sigma}^{(\lambda, \lambda) }\Phi_{\lambda} ^*\right)+\mathcal{O}(\beta^2).
\]
Since $\operatorname{Im}\Phi_\lambda=\ker(U_0-\lambda I)$
and
\[
\mathcal{H}_0
=
\left(
\ker(U_0-I)\oplus\ker(U_0+I)
\right)^\perp,
\]
we have
\[
\Phi_\lambda^*|_{\mathcal{H}_0}=0,
\qquad \lambda\in\{\pm1\}.
\]
Hence,
\[
(e^{-i\beta\tau H}-U_\beta^\tau)|_{\mathcal{H}_0}
=
\mathcal{O}(\beta^2)
\]
in the sense of the definition above. Since $G$ is finite,
the equivalence of matrix norms implies
\[
\left\|(e^{-i\beta\tau H}-U_\beta^\tau)|_{\mathcal{H}_0}\right\|
=
\mathcal{O}(\beta^2).
\]
}
\end{proof}

We are now in a position to prove Theorem~\ref{main} and Theorem~\ref{main2}.

\begin{proof}[Proof of Theorem~\ref{main}]

{
Set $n_\beta=\left\lceil t/\beta\right\rceil$.
By the triangle inequality, we obtain 
\begin{align*}
&\|\psi_t^{[C]}-\varphi_{n_\beta}^{[D]}\| \\
&\leq 
\|\left(e^{-i\beta\tau H}\right)^{n_\beta}\phi_0- \left(U_\beta^\tau\right)^{n_\beta}\phi_0 \|
+
\|\psi_t^{[C]}- \left(e^{-i\beta\tau H}\right)^{n_\beta}\phi_0 \| +
 \|\varphi_{n_\beta}^{[D]}- \left(U_\beta^\tau\right)^{n_\beta}\phi_0 \| \\
 &=
 \|\left(e^{-i\beta\tau H}\right)^{n_\beta}\phi_0- \left(U_\beta^\tau\right)^{n_\beta}\phi_0 \|
+
\|\psi_t^{[C]}- \left(e^{-i\beta\tau H}\right)^{n_\beta}\phi_0 \| .
\end{align*}
Let us estimate each terms.
\begin{enumerate}
\renewcommand{\labelenumi}{(\roman{enumi})}
    \item 
Estimation of the first term: 
We first use the identity
\begin{align}\label{tel}
\left(\left(e^{-i\beta\tau H}\right)^{n_\beta}-\left(U_\beta^\tau\right)^{n_\beta}\right)\phi_0
=
\sum_{j=0}^{n_\beta-1}
\left(e^{-i\beta\tau H}\right)^{n_\beta-1-j}
(e^{-i\beta\tau H}-U_\beta^\tau)
\left(U_\beta^\tau\right)^j\phi_0.
\end{align}
Since $e^{-i\beta\tau H}$ and $U_\beta^\tau$ are unitary,
\[
\left\|\left(\left(e^{-i\beta\tau H}\right)^k\right)|_{\mathcal{H}_0}\right\|
\leq
\left\|\left(e^{-i\beta\tau H}\right)^k\right\|=1
\]
and
\[
\left\|\left(\left(U_\beta^\tau\right)^k\right)|_{\mathcal{H}_0}\right\|
\leq
\left\|\left(U_\beta^\tau\right)^k\right\|=1
\]
Taking the norm of both sides of \eqref{tel}, and using
$\phi_0\in\mathcal H_0$ and
$(U_\beta^\tau)^j\phi_0\in\mathcal H_0$ for every nonnegative integer $j$,
we obtain
\begin{align*}
\left\|\left(\left(e^{-i\beta\tau H}\right)^{n_\beta}-\left(U_\beta^\tau\right)^{n_\beta}\right)\phi_0\right\|
&=
\sum_{j=0}^{n_\beta-1}
\left(e^{-i\beta\tau H}\right)^{n_\beta-1-j}
(e^{-i\beta\tau H}-U_\beta^\tau)
\left(U_\beta^\tau\right)^j\phi_0 \\
&\leq
n_\beta\|(e^{-i\beta\tau H}-U_\beta^\tau)|_{\mathcal{H}_0}\|\|\phi_0\|.
\end{align*}
By Corollary~\ref{beta^2}, there exists a constant $C>0$ for sufficiently small $\beta$
such that
\begin{align*}
\left\|\left(\left(e^{-i\beta\tau H}\right)^{n_\beta}-\left(U_\beta^\tau\right)^{n_\beta}\right)\phi_0\right\|
&\leq
n_\beta\|(e^{-i\beta\tau H}-U_\beta^\tau)|_{\mathcal{H}_0}\|\|\phi_0\|\\
&\leq
Cn_\beta\beta^2\|\phi_0\| \\
&<
C(t\beta +\beta^2)\|\phi_0\| 
\\
&\longrightarrow0,
\end{align*}
where the last inequality follows from
\[
\frac{t}{\beta}
\leq
n_\beta
=
\left\lceil\frac{t}{\beta}\right\rceil
<
\frac{t}{\beta}+1.
\]
\item Estimation of the second term:
For an eigenvalue $\nu$ of $H$, let $P_{\nu}^{(H)}$ denote the corresponding eigenprojection. Then, we obtain
\begin{align*}
    \|\psi_t^{[C]}- \left(e^{-i\beta\tau H}\right)^{n_\beta}\phi_0 \|^2 
    &=
    \|\left(e^{-it\tau H}-e^{-i\beta n_{\beta}\tau H}\right)\phi_{0}\|^2 \\
    &=
    \|e^{-it\tau H}\left(I-e^{i(t-\beta n_{\beta})\tau H}\right)\phi_{0}\|^2 \\
    &\leq
    \|\left(I-e^{i(t-\beta n_{\beta})\tau H}\right)\phi_{0}\|^2 \\
    &=
    \left\|\sum_{\nu \in \Spec(H)}\left( 1-e^{i(t-\beta n_{\beta})\tau \nu}\right)P_\nu^{(H)}\phi_{0}\right\|^2 \\
    &=
    \sum_{\nu \in \Spec(H)}\left\|\left( 1-e^{i(t-\beta n_{\beta})\tau \nu}\right)P_\nu^{(H)}\phi_{0}\right\|^2 \\
    &=
    \sum_{\nu \in \Spec(H)}| 1-e^{i(t-\beta n_{\beta})\tau \nu}|^2\left\|P_\nu^{(H)}\phi_{0}\right\|^2 \\
    &\leq
    \sum_{\nu \in \Spec(H)}| 1-e^{i(t-\beta n_{\beta})\tau \nu}|^2 \\
    &=
    \sum_{\nu \in \Spec(H)}
    4\sin^2{\frac{(t-\beta n_{\beta})\tau \nu}{2}}
\end{align*}
where the fourth equality follows from the spectral decomposition of $H$, and the fifth equality follows from the orthogonality of the eigenspaces
of the Hermitian matrix $H$. Using $|\sin x|\leq |x|$, we obtain
\begin{align*}
\|\psi_t^{[C]}- \left(e^{-i\beta\tau H}\right)^{n_\beta}\phi_0 \|^2
&\leq
(t-\beta n_\beta)^2\tau^2
\sum_{\nu\in\Spec(H)}\nu^2.
\end{align*}
Since $n_\beta=\lceil t/\beta\rceil$, we have
$$
0\leq \beta n_\beta-t<\beta,
$$
and hence
$$
\|\psi_t^{[C]}- \left(e^{-i\beta\tau H}\right)^{n_\beta}\phi_0 \|^2
\leq
\beta^2\tau^2
\sum_{\nu\in\Spec(H)}\nu^2 = O(\beta^2).
$$
Therefore,
$$
\|\psi_t^{[C]}- \left(e^{-i\beta\tau H}\right)^{n_\beta}\phi_0 \|
\longrightarrow 0.
$$
\end{enumerate}
Combining (i) and (ii), we obtain
\[
\|\psi_t^{[C]}-\varphi_{n_\beta}^{[D]}\|
\longrightarrow 0.
\]
Thus, this proves the claim.}
    
\end{proof}

\begin{proof}[Proof of Theorem~\ref{main2}]
We first prove \eqref{mu} and \eqref{muv}. After that, we establish
\eqref{kerh}.
By \eqref{H1 2}, we have
     \[
     H
           =
     \sum_{\lambda\in \Spec(U_0) \setminus \{\pm1\}} 
           \Phi_{\lambda} M_{\sigma}^{(\lambda, \lambda)} \Phi_{\lambda} ^*.
     \]
     Note that $M_\sigma^{(\lambda,\lambda)} = O$ for $\lambda$ with  $\dim \left( \ker (T- \lambda_T I) \cap \mathcal{B}_{\partial X_a}\right) \leq 1$.
     Thus, only the eigenvalues \(\lambda\) satisfying
\[
\dim \left( \ker (T- \lambda_T I) \cap \mathcal{B}_{\partial X_a}\right)=2
\]
contribute to \(H\). We therefore introduce 
\begin{align*}
     &\Spec^{(2)}(U_0) \\
     &:= \Bigl\{ \lambda = e^{\pm i \arccos \lambda_T}\in\Spec(U_0)\setminus\{\pm1\} \mid \dim \bigl(\ker(T-\lambda_T I)\cap\mathcal B_{\partial X_a}\bigr)=2 \Bigr\}. 
     \end{align*}
Then, by \eqref{SD},
\begin{align}
H
&=
\sum_{\substack{
\lambda \in \Spec^{(2)}(U_0) 
}}
\Phi_{\lambda}
M_{\sigma}^{(\lambda,\lambda)}
\Phi_{\lambda}^{*} \notag \\
&=
\sum_{\substack{
\lambda \in \Spec^{(2)}(U_0) 
}}
\frac{i E_{\lambda_T}(a, a)}{\sqrt{1- \lambda_T^2}}
    \begin{bmatrix}
    \ket{\lambda^{(1)}} &\ket{\lambda^{(2)}}
    \end{bmatrix}
    \begin{bmatrix}
        0 & 1 \\
        -1 & 0
    \end{bmatrix}
    \begin{bmatrix}
    \bra{\lambda^{(1)}} \\
    \bra{\lambda^{(2)}}
    \end{bmatrix}  \label{Hsum} \\
&=   
\sum_{\substack{
\lambda \in \Spec^{(2)}(U_0) 
}}
    \frac{- E_{\lambda_T}(a, a)}{\sqrt{1- \lambda_T^2}}
    \frac{1}{\sqrt{2}}\left( \ket{\lambda^{(1)}} + i\ket{\lambda^{(2)}} \right)
    \frac{1}{\sqrt{2}}
    \left( \bra{\lambda^{(1)}} - i\bra{\lambda^{(2)}} \right) \notag \\
    & \qquad
    +\negthickspace
    \sum_{\substack{
\lambda \in \Spec^{(2)}(U_0) 
}}
    \frac{ E_{\lambda_T}(a, a)}{\sqrt{1- \lambda_T^2}}
    \frac{1}{\sqrt{2}}\left( \ket{\lambda^{(1)}} - i\ket{\lambda^{(2)}} \right)
    \frac{1}{\sqrt{2}}
    \left( \bra{\lambda^{(1)}} + i\bra{\lambda^{(2)}} \right), \notag
\end{align}
where the third equality follows from the spectral decomposition of
\[
\begin{bmatrix}
0 & 1\\
-1 & 0
\end{bmatrix}
=
i 
\begin{bmatrix}
\frac{1}{\sqrt{2}} \\
\frac{i}{\sqrt{2}} 
\end{bmatrix}
\begin{bmatrix}
\frac{1}{\sqrt{2}} &
-\frac{i}{\sqrt{2}} 
\end{bmatrix}
-i 
\begin{bmatrix}
\frac{1}{\sqrt{2}} \\
-\frac{i}{\sqrt{2}} 
\end{bmatrix}
\begin{bmatrix}
\frac{1}{\sqrt{2}} &
\frac{i}{\sqrt{2}} 
\end{bmatrix},
\]
Hence, we obtain a spectral decomposition of \(H\). The orthonormality of 
$\{ \frac{1}{\sqrt{2}}\left( \ket{\lambda^{(1)}} \pm i\ket{\lambda^{(2)}} \right)\}$
follows immediately from the fact that $\{\ket{\lambda^{(1)}},\ket{\lambda^{(2)}}\}$ 
is an orthonormal basis of \(\ker(U_0-\lambda I)\). By Lemma~\ref{GSMT},
\[
\ket{\lambda^{(1)}}=\partial_\lambda^* f_1^{(\lambda_T)},
\qquad
\ket{\lambda^{(2)}}=\partial_\lambda^* f_2^{(\lambda_T)}.
\]
Thus, \eqref{mu} and \eqref{muv} are follow.

Next, we prove \eqref{kerh}.
Since the summation in \eqref{Hsum} is taken over all $\lambda$ satisfying
\[
\dim\bigl(\ker(T-\lambda_T I)\cap\mathcal B_{\partial X_a}\bigr)=2, 
\]
we have
\[
H \ket{\psi} =0
\]
for any $\ket{\psi} \in \mathcal{S}_{\mathrm{sim}} 
    \oplus
    \mathcal{T}_{\text{per}}.$ 
Moreover, since the summation in \eqref{Hsum} is taken over all $\lambda \neq \pm1$,  we have
\[
H \ket{\psi} =0
\]
for any $\ket{\psi} \in \mathcal{L}^{\perp}$.
Therefore, 
\begin{align}\label{eq:in}
\mathcal{S}_{\mathrm{sim}} 
    \oplus
    \mathcal{T}_{\text{per}}
    \oplus
    \mathcal{L}^{\perp} \subset \ker H.
\end{align}
By Lemma~\ref{GSMT}, the summation index in \eqref{Hsum}, and definitions of $\mathcal{S}_{\mathrm{sim}}$ and $\mathcal{T}_{\text{per}}$, we have
\[
\dim \mathcal{L}
=
 \dim \mathcal{S}_{\mathrm{sim}} + \dim \mathcal{T}_{\text{per}} +
 \operatorname{rank}(H).
\]
 Hence, 
\[
 |A(G)|
 =
 \dim \mathcal{S}_{\mathrm{sim}} + \dim \mathcal{T}_{\text{per}} +
 \operatorname{rank}(H) + \dim \mathcal{L}^{\perp}.
\]
On the other hand, by the rank--nullity theorem,
\[
|A(G)|
=
\operatorname{rank}(H)
+
\dim \ker H.
\]
Comparing the RHS of the above equalities, we obtain

\[
\dim \ker H = \dim \mathcal{S}_{\mathrm{sim}} 
    +
    \dim \mathcal{T}_{\text{per}}
    +
    \dim\mathcal{L}^{\perp}.
\]
Together with \eqref{eq:in}, this implies that
\[
    \ker H = \mathcal{S}_{\mathrm{sim}} 
    \oplus
    \mathcal{T}_{\text{per}}
    \oplus
    \mathcal{L}^{\perp}.
    \] 
    This completes the proof of \eqref{kerh}.
    \end{proof}

\section{Several examples of $H$}

In this section, we investigate several concrete examples of \(H\) to clarify its structure.
Since \(H\) is characterized by eigenvalues of multiplicity at least two, we focus on three representative periodic graphs: the cycle graph \(C_n\), the complete bipartite graph \(K_{k,k}\), and the path graph \(P_n\),
whose spectral multiplicities exhibit distinct patterns.

\subsection{Example 1: The cycle graph $C_n$ }

First, we examine $H$ associated with the cycle graph $C_n$. $C_n$ is known to induce an $n$-periodic Grover walk~\cite{SHH}. Furthermore, all eigenvalues of its adjacency matrix, except for $2$ (and $-2$ when $n$ is even), are of multiplicity two~\cite{BH}.

We label the vertices of $C_n$ by $0,1,\dots,n-1$, 
    and order the components of eigenvectors accordingly. 
    We divide the set of directed edges into two classes $C_+$ and $C_-$, defined by
\[
C_+ = \{ (u,v) \mid v \equiv u+1 \ (\mathrm{mod}\ n) \}, \qquad
C_- = \{ (u,v) \mid v \equiv u-1 \ (\mathrm{mod}\ n) \}.
\]
\begin{pro}
    Let $G$ be the $n$-vertex cycle graph $C_n$. For a strength {$\beta$} of a vector potential and a reference edge $a \in A(G)$, set $U_{\beta} = U_{\beta}(G,a)$. Here, we assume that $0 < \beta \ll 1 $. Then,
\[
H {=}
\begin{bmatrix}
1 & 0\\
0 & -1
\end{bmatrix}
\otimes {H'},
\]
where the matrix \({H'}\), indexed by \(A^{(C_+)}\), is given by the following entries for \(e_1,e_2\in C_+\).

If $n$ is odd, 
\[
({H'})_{e_1,e_2}
=
\begin{cases}
\displaystyle \frac{n-2}{n^2}
& e_1=e_2
,\\[6pt]
\displaystyle -\frac{2}{n^2}
& \text{otherwise},\\[6pt]
\end{cases}
\]
and if $n$ is even, 
\[
({H'})_{e_1,e_2}
=
\begin{cases}
\displaystyle \frac{n-2}{n^2}
& e_1=e_2
,\\[6pt]
\displaystyle -\frac{2}{n^2}
& e_1 \neq e_2,\ o(e_1)-o(e_2)\equiv 0 \pmod 2,\\[6pt]
\displaystyle 0
& \text{otherwise}.
\end{cases}
\]
for $e_1,e_2\in C_+.$
\end{pro}

\begin{figure}[htbp]
    \centering
    \includegraphics[width=\textwidth]{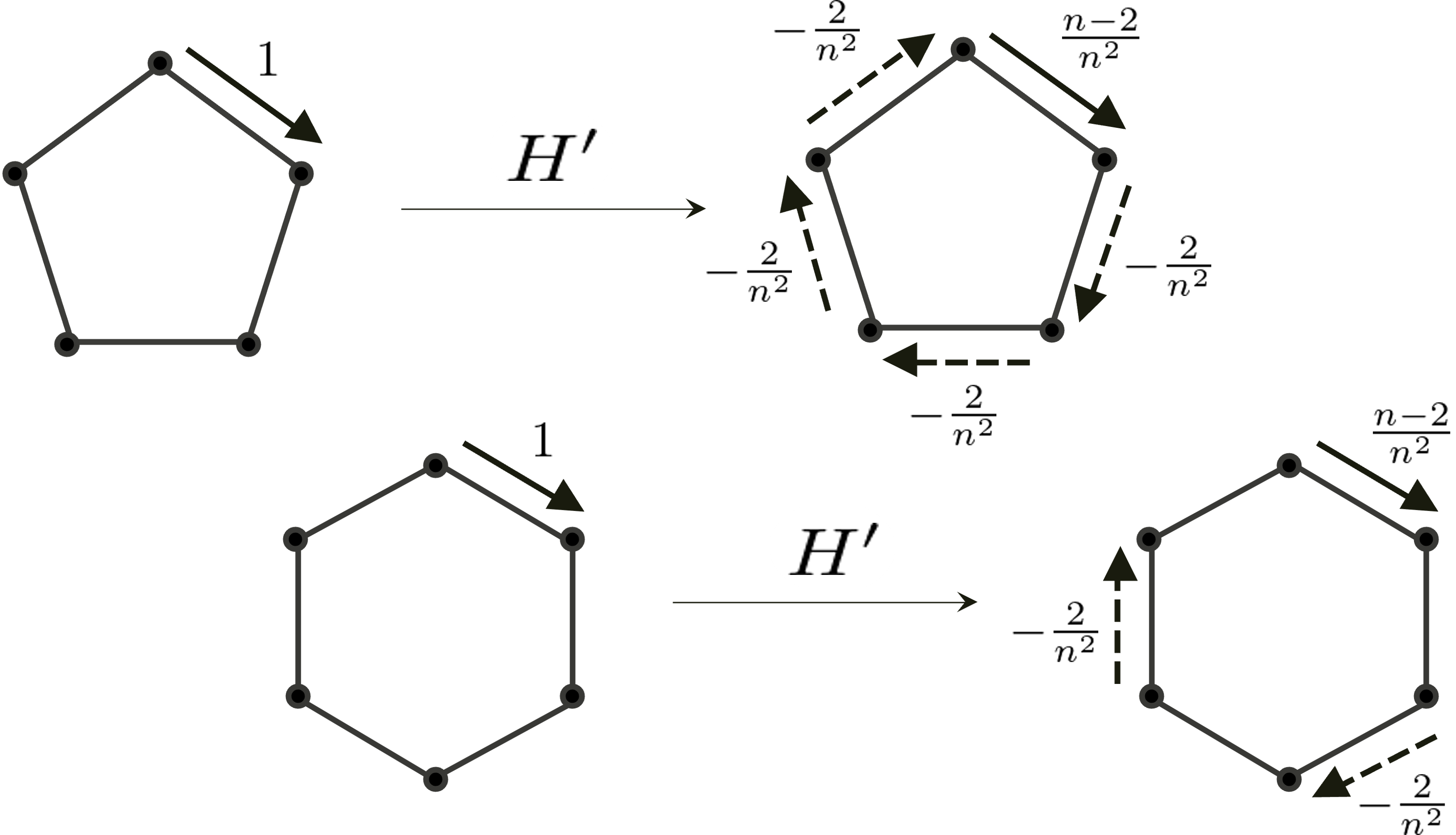}
    \caption{
        ${H'}$ on $C_5$ : 
        The left figure represents ${H'} \delta_i= \frac{3}{25}\delta_i -\frac{2}{25}\left(
        \delta_{i+1}+
        \delta_{i+2}+
        \delta_{i+3}+
        \delta_{i+4}
        \right)$
     \ ($ i \in \mathbb{Z}_5$).
        ${H'}$ on $C_6$ : 
        The right figure represents ${H'} \delta_j= \frac{1}{9}\delta_j -\frac{1}{18}\left(
        \delta_{j+2}+
        \delta_{j+4}
        \right)$
     \ ($ j \in \mathbb{Z}_6$).
    }
    \label{fig:c}
\end{figure}

\begin{proof}
    
Since $C_n$ is vertex-transitive, we may assume without loss of generality that the reference edge $a$ satisfies $o(a)=0$ and $t(a)=1$. 
The eigenvalues of \(T(C_n)\) with multiplicity greater than one are
\[
{
\mu_k = \cos\left(\frac{2\pi k}{n}\right),
\qquad
k=1,\dots,\left\lceil \frac{n}{2}-1 \right\rceil.}
\]
For each $\mu_k$, the eigenvectors $f_1^{(\mu_k)}$ and $f_2^{(\mu_k)}$ chosen as the above are real-valued and linearly independent, and their supports intersect $\partial X_a$. They are given by
\[
f_1^{(\mu_k)}
= 
\sqrt{\frac{2}{n}}
\begin{bmatrix}
1 \\
\cos\left(\frac{2\pi k}{n}\right) \\
\vdots \\
\cos\left(\frac{2\pi(n-1)k}{n}\right)
\end{bmatrix},
\quad
f_2^{(\mu_k)}
= 
\sqrt{\frac{2}{n}}
\begin{bmatrix}
0 \\
\sin\left(\frac{2\pi k}{n}\right) \\
\vdots \\
\sin\left(\frac{2\pi(n-1)k}{n}\right)
\end{bmatrix},
\]
where $f_1^{(\mu_k)}$ and $f_2^{(\mu_k)}$ are defined as in \eqref{f}.
Then, 
\begin{align*}
    E_{\mu_k}(a, a)
    &= 
    \frac{1}{2}
    \begin{vmatrix}f_1^{(\mu_k)} (o(a)) &
    f_2^{(\mu_k)} (o(a)) \\
    f_1^{(\mu_k)} (t(a)) &
    f_2^{(\mu_k)} (t(a))\end{vmatrix} \\
    &=
    \frac{1}{n}
    \sin\left(\frac{2\pi k}{n}\right).
\end{align*}
Moreover, 
\begin{align*}
    &E_{\mu_k}(e_1, e_2)
    +
    E_{\mu_k}(e_2, e_1) \\
    &= 
    \frac{1}{n}
    \sin\left(\frac{2\pi k}{n}\left(t(e_2)-o(e_1)\right)\right)
    +
    \frac{1}{n}
    \sin\left(\frac{2\pi k}{n}\left(t(e_1)-o(e_2)\right)\right) \\
    &=
    \frac{2}{n}
    \sin\left(\frac{\pi k}{n}\left(t(e_1)+
    t(e_2)-o(e_1)-o(e_2)\right)\right) \\ 
    & \qquad \qquad \times \cos\left(\frac{\pi k}{n}\left(t(e_1)-
    t(e_2)+
    o(e_1)-o(e_2)\right)\right).
\end{align*}
If $e_1 \in C_{\pm}$ and $e_2 \in C_{\mp}$, then
\[
\sin\left(\frac{\pi k}{n}\left(t(e_1)+
    t(e_2)-o(e_1)-o(e_2)\right)\right) = 0.
\]
If $e_1 \in C_{\pm}$ and $e_2 \in C_{\pm}$, then
\[
\sin\left(\frac{\pi k}{n}\left(t(e_1)+
    t(e_2)-o(e_1)-o(e_2)\right)\right)
    =
    \pm\sin\left(\frac{2\pi k}{n}\right)
\]
and
\[
\cos\left(\frac{\pi k}{n}\left(t(e_1)-
    t(e_2)+
    o(e_1)-o(e_2)\right)\right) 
    =
    \cos\left(\frac{2\pi k}{n}\left(
    o(e_1)-o(e_2)\right)\right).
\]
Therefore,
\begin{align*}
    (H)_{e_1, e_2} 
    &=
    \begin{cases}
    \displaystyle
    \frac{2}{n^2}\sum_{k=1}^{\lceil \frac{n}{2}-1\rceil} \cos\left(\frac{2\pi k}{n}\left(
    o(e_1)-o(e_2)\right)\right) & e_1, e_2 \in C_+, \\
    \displaystyle
    -\frac{2}{n^2}\sum_{k=1}^{\lceil \frac{n}{2}-1\rceil } \cos\left(\frac{2\pi k}{n}\left(
    o(e_1)-o(e_2)\right)\right) & e_1, e_2 \in C_-.
    \end{cases}
\end{align*}
Thus,
\[
H \cong
\begin{bmatrix}
1 & 0\\
0 & -1
\end{bmatrix}
\otimes H',
\]
where
\[
(H')_{e_1,e_2}
=
\frac{2}{n^2}
\sum_{k=1}^{\left\lceil \frac{n}{2}-1 \right\rceil}
\cos\left(
\frac{2\pi k}{n}
\bigl(
o(e_1)-o(e_2)
\bigr)
\right)
\]
for $e_1,e_2 \in C_+$.
Put
\[
o(e_1)-o(e_2)=l.
\]
If $l \equiv 0 \pmod n$, then
\[
(H')_{e_1,e_2}
=
\frac{n-2}{n^2}.
\]
Assume that $l \not\equiv 0 \pmod n$.
We distinguish two cases according to the parity of $n$.

\medskip
\noindent
\textbf{Case 1.} \(n\) is even, say \(n=2m\). 

\noindent

\begin{align*}
    \left(H'\right)_{e_1,e_2}
&=
\frac{2}{n^2}
\sum_{k=1}^{m-1}
\cos\left(
\frac{2\pi lk}{2m}
\right) \\
&=
\frac{1}{n^2}
\sum_{k=1}^{m-1}
\left(
e^{\frac{2\pi lk}{2m}i} + e^{\frac{-2\pi lk}{2m}i}
\right) \\
&=
\frac{\Re \left(
-1 - e^{l\pi i} + e^{\frac{2\pi l}{2m}i} + e^{l\pi i} e^{-\frac{2\pi l}{2m}i}
\right)}{n^2\left(1- \cos\left(\frac{2\pi l}{2m}\right)\right)}.
\end{align*}
It follows that
\[
(H')_{e_1,e_2}=0
\]
when $l$ is odd, and
\[
(H')_{e_1,e_2}=-\frac{2}{n^2}
\]
when $l$ is even. Hence, the assertion follows.

\medskip
\noindent
\textbf{Case 2.} \(n\) is odd, say \(n=2m+1\). 

\noindent
In the similar fashion to Case 1, we obtain
\[
\left(H'\right)_{e_1,e_2}
=
-\frac{2}{n^2} + \frac
{2}
{n^2(1-\cos{\frac{2\pi l}{n}})}
\Re\left(
e^{\frac{2\pi ml}{2m+1}i}-e^{\frac{2\pi(m+1)l}{2m+1}i}
\right).
\]
Moreover, we find that
\begin{align*}
    \Re\left(
e^{\frac{2\pi ml}{2m+1}i}-e^{\frac{2\pi(m+1)l}{2m+1}i}
\right) 
&= 
-e^{l\pi i}
\Re\left(
e^{\frac{l\pi }{2m+1}i}-e^{\frac{-l\pi }{2m+1}i}
\right) \\
&=
-2e^{l\pi i}
\Re\left(
i\sin\left(\frac{l\pi}{2m+1}
\right)\right) \\
&=
0.
\end{align*}
Substituting this into the above expression, we obtain
\[
(H')_{e_1,e_2}=-\frac{2}{n^2}.
\]
Hence, the assertion follows.
\end{proof}

\subsection{Example 2: The complete bipartite graph $K_{k, k}$ }

Second, we examine \(H\) associated with the complete bipartite graph \(K_{k,k}\). 
\(K_{k,k}\) is known to induce {a} $4$-periodic Grover walk~\cite{YNIE}. Furthermore, {among} the eigenvalues of the adjacency matrix, only \(0\) has multiplicity greater than one~\cite{BH}.

Let \(K_{k,k}\) be the complete bipartite graph with bipartition
\((V_1,V_2)\), where
\[
V_1=\{x_1,\dots,x_k\}, \qquad
V_2=\{y_1,\dots,y_k\}.
\]
For \(j=1,2\), we denote by
\[
A^{(V_j)} = A^{(V_j)}(K_{k,k})
=
\{\,e\in A(K_{k,k}) \mid o(e)\in V_j\,\}
\]
the set of directed edges with origin in \(V_j\).

\begin{pro}
    Let $G$ be the complete bipartite graph $K_{k,k}$. For a strength {$\beta$} of a vector potential  and a reference edge $a \in A(G)$, set $U_{\beta} = U_{\beta}(G,a)$. Here, we assume that $0 < \beta \ll 1 $.Then,
\[
H {=}
\begin{bmatrix}
1 & 0\\
0 & -1
\end{bmatrix}
\otimes {H'},
\]
where the matrix \({H'}\), indexed by \(A^{(V_1)}\), is given by
\[
({H'})_{(x_i,y_j),(x_s,y_t)}
=
\begin{cases}
\displaystyle \frac{2(k-1)^2}{k^4}
& \text{if all of } i,j,s,t \text{ are equal to } k,\\[6pt]
\displaystyle \frac{(k-1)(k-2)}{k^4}
& \text{if exactly three of } i,j,s,t \text{ are equal to } k,\\[6pt]
\displaystyle -\frac{2(k-1)}{k^4}
& \text{if exactly two of } i,j,s,t \text{ are equal to } k,\\[6pt]
\displaystyle \frac{-k+2}{k^4}
& \text{if exactly one of } i,j,s,t \text{ is equal to } k,\\[6pt]
\displaystyle \frac{2}{k^4}
& \text{if none of } i,j,s,t \text{ is equal to } k.
\end{cases}
\]
\end{pro}

\begin{figure}[htbp]
    \centering
    \includegraphics[width=\textwidth]{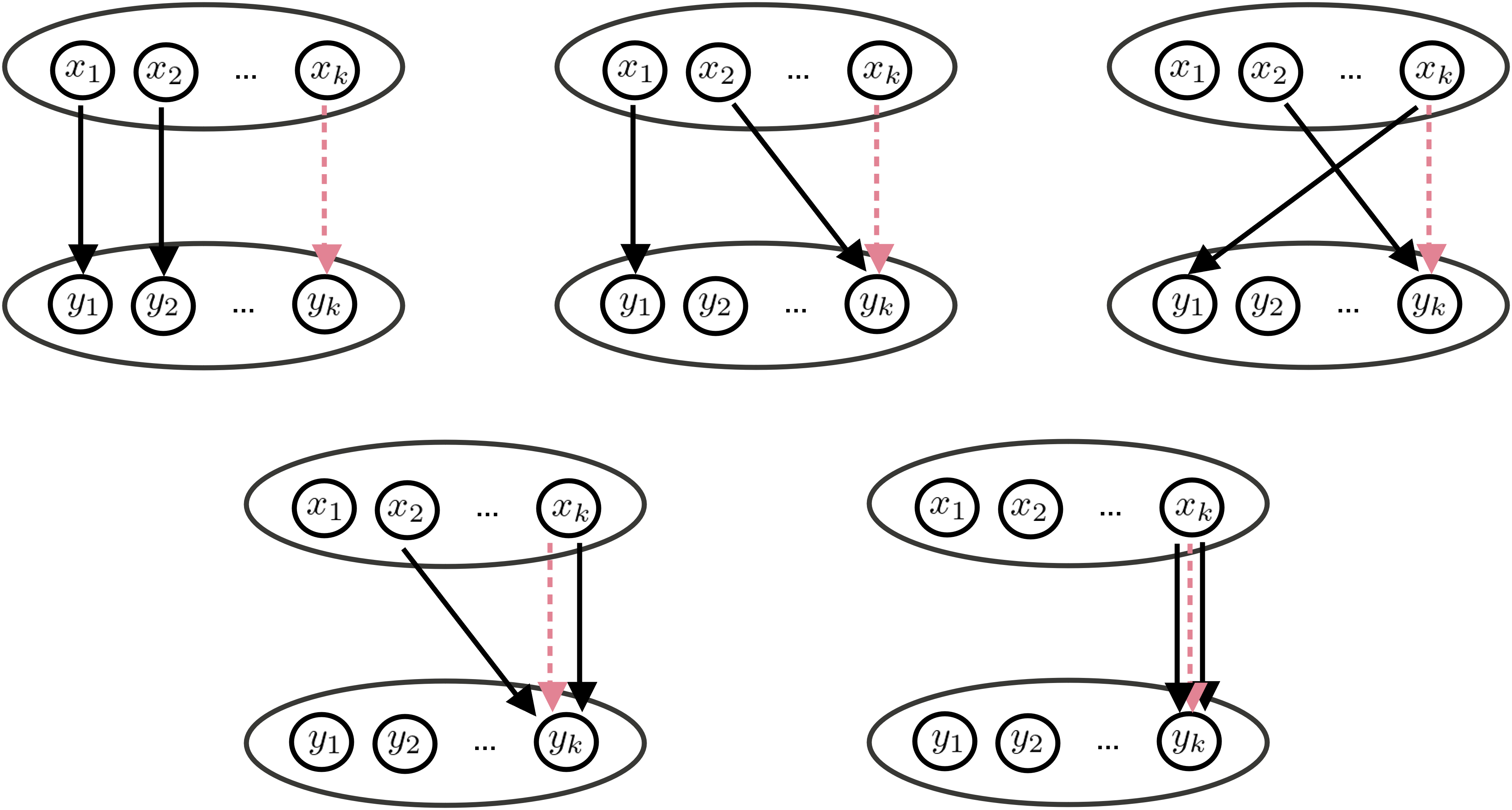}
    \caption{
        Configurations of directed edges corresponding to the entries of ${H'}$: The five configurations shown above correspond to the five cases in the definition of ${H'}$. The two solid directed edges represent the indices of an entry of ${H'}$, while the dashed directed edge represents the reference edge $a$.
    }
    \label{fig:all}
\end{figure}

\begin{figure}[htbp]
    \centering
    \includegraphics[width=\textwidth]{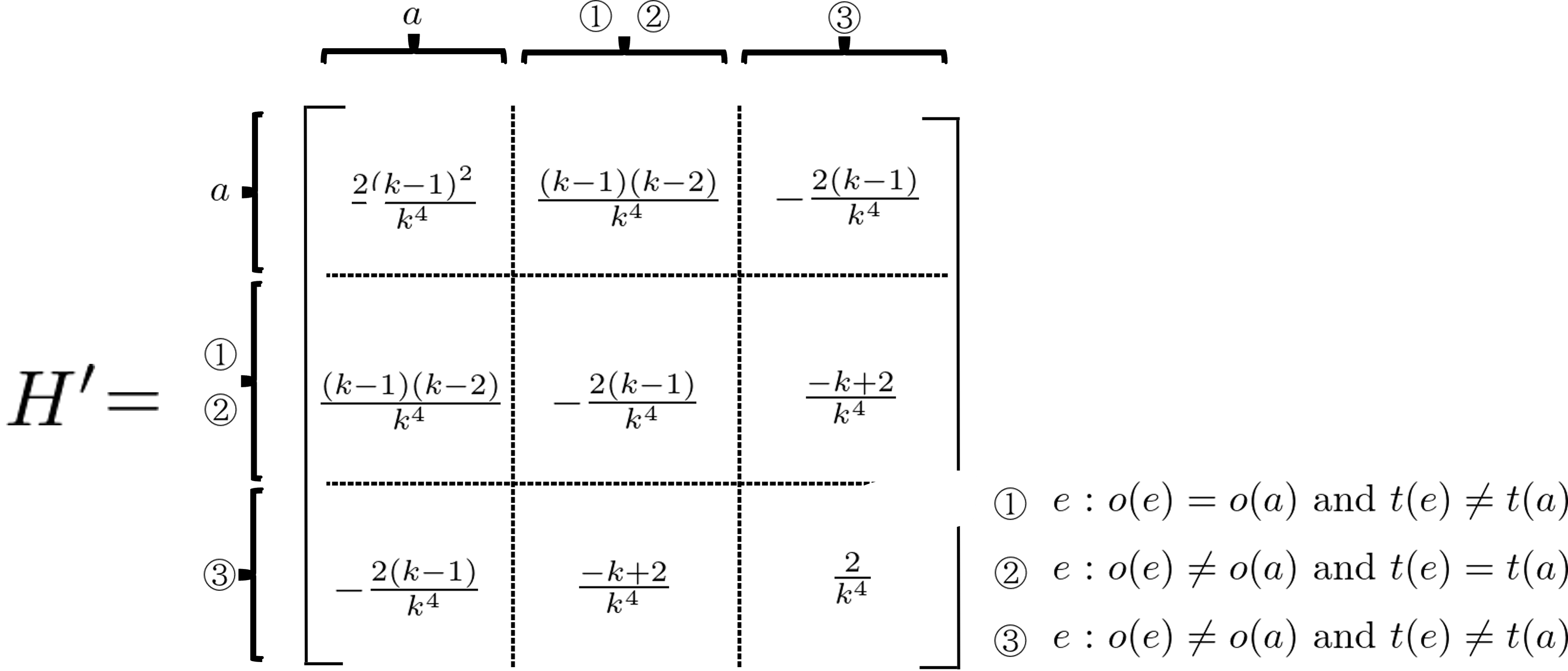}
    \caption{
        Values of the entries of \({H'}\) corresponding to each edge configuration.
    }
    \label{fig:mat}
\end{figure}

\begin{proof}
    Without loss of generality, we may assume that
\[
o(a)=x_k, \qquad t(a)=y_k,
\]
and label the rows and columns of the discriminant matrix \(T\) by
\[
x_1,\dots,x_k,y_1,\dots,y_k
\]
in this order. 
 The only eigenvalue of \(T\) with multiplicity greater than \(1\) is \(0\).
An orthonormal basis of the corresponding eigenspace is given by
\[
\frac{1}{\sqrt{2}}
\begin{bmatrix}
1\\
-1\\
0\\
0\\
\vdots\\
0\\
\mathbf{0}_{k}
\end{bmatrix}, 
\frac{1}{\sqrt{6}}
\begin{bmatrix}
1\\
1\\
-2\\
0\\
\vdots\\
0\\
\mathbf{0}_{k}
\end{bmatrix}, 
\dots, 
\frac{1}{\sqrt{k(k-1)}}
\begin{bmatrix}
1\\
1\\
1\\
1\\
\vdots\\
-(k-1)\\
\mathbf{0}_{k}
\end{bmatrix},
\]
and 
\[
\frac{1}{\sqrt{2}}
\begin{bmatrix}
\mathbf{0}_{k}\\
1\\
-1\\
0\\
0\\
\vdots\\
0
\end{bmatrix}, 
\frac{1}{\sqrt{6}}
\begin{bmatrix}
\mathbf{0}_{k}\\
1\\
1\\
-2\\
0\\
\vdots\\
0
\end{bmatrix}, 
\dots, 
\frac{1}{\sqrt{k(k-1)}}
\begin{bmatrix}
\mathbf{0}_{k}\\
1\\
1\\
1\\
1\\
\vdots\\
-(k-1)
\end{bmatrix},
\]
where $\mathbf{0}_{k}$ denotes the zero vector in $\mathbb{C}^k$.
From this basis, we choose
\[
f_1^{(0)}
=
\frac{1}{\sqrt{k(k-1)}}
\begin{bmatrix}
1\\
\vdots\\
1\\
-(k-1)\\
\mathbf{0}_{k}
\end{bmatrix},
\quad
f_2^{(0)}
= 
\frac{1}{\sqrt{k(k-1)}}
\begin{bmatrix}
\mathbf{0}_{k}\\
1\\
\vdots\\
1\\
-(k-1)
\end{bmatrix},
\]
where $f_1^{(0)}$ and $f_2^{(0)}$ are defined as in \eqref{f}. Then,
\begin{align*}
    E_0(a,a)
    =
    \frac{1}{k}
    \begin{vmatrix}f_1^{(0)} (x_k) &
    f_2^{(0)} (x_k) \\
    f_1^{(0)} (y_k) &
    f_2^{(0)} (y_k)\end{vmatrix} 
    =
    \frac{k-1}{k^2}.
\end{align*}
We next evaluate \(E_0(e_1,e_2)\) according to the partite sets containing
\(o(e_1)\) and \(o(e_2)\).
Suppose that \(o(e_1)\in V_1\) and \(o(e_2)\in V_2\).
Since
\[
f_1^{(0)}(t(e_1))
=
f_2^{(0)}(t(e_2))
=
f_1^{(0)}(o(e_2))
=
f_2^{(0)}(o(e_1))
=
0,
\]
it follows that
\[
E_0(e_1,e_2)=E_0(e_2,e_1)=0.
\]
Similarly, if \(o(e_1)\in V_2\) and \(o(e_2)\in V_1\), then
\[
E_0(e_1,e_2)=E_0(e_2,e_1)=0.
\]
The remaining cases, where \(o(e_1)\) and \(o(e_2)\) belong to the same partite set, are handled similarly.
By distinguishing whether \(e_1\) and \(e_2\) share endpoints with the reference edge \(a\), the statement of the theorem follows.

\end{proof}

\subsection{Example 3: The path graph $P_n$ }

Finally, we consider \(H\) associated with the path graph \(P_n\). $P_n$ is known to induce {a} $2(n-1)$-periodic Grover walk~\cite{SHH}.
Since all eigenvalues of \(P_n\) are simple \cite{BH}, the following proposition holds.

\begin{pro}
    Let $G$ be the $n$-vertex path $P_n$. For {a} strength {$\beta$} of a vector potential and a reference edge $a \in A(G)$, set $U_{\beta} = U_{\beta}(G,a)$. Here, we assume that $0 < \beta \ll 1 $. Then,
    \[
    ||U_{\beta}^{2n}-I|| = \mathcal{O}(\beta^2).
    \]
\end{pro}
\begin{proof}
    From \cite{SHH}, the period of the path $P_n$ is ${2(n-1)}$.
    In addition,  Lemma~\ref{GSMT}, any eigenvalue of $U_0$ has multiplicity one. Let $\ket{\lambda}$ be an eigenvector of $U_0$ associated with the eigenvalue $\lambda$. 
    Then, 
    \begin{align*}
        H
        &=
        2n \sum_{\lambda\in \Spec(U_{0})} \Phi_{\lambda} M_{\sigma}^{(\lambda, \lambda)} \Phi_{\lambda} ^* \\
        &=
        2n \sum_{\lambda\in \Spec(U_{0})} \ket{\lambda}
        \bra{\lambda}\sigma\ket{\lambda}
        \bra{\lambda}
        \\
        &=
        0.
    \end{align*}
    Here, the third equality follows from Lemma~\ref{braket}. 
\end{proof}

\section*{Acknowledgements}
E.S. acknowledges financial support from the Grant-in-Aid of Scientific Research (C) JSPS KAKENHI Grant No.~24K06863. 

\section*{Data availability}
No datasets were generated or analyzed during the current study.

\end{document}